\documentclass[11pt]{article}
\usepackage{graphicx}
\usepackage{pgfplots}
\usepackage[dvipsnames]{xcolor}
\usepackage[colorlinks=true,citecolor=blue]{hyperref}
\usepackage{natbib}
\usepackage{graphicx}
\usepackage{amsfonts}
\usepackage{amsmath}
\usepackage{amssymb}
\usepackage{url}
\usepackage{fancyhdr}
\usepackage{indentfirst}
\usepackage{enumerate}
\usepackage{titlesec}
\usepackage{amsthm}
\usepackage{dsfont}
\usepackage[misc]{ifsym}
\usepackage{enumitem}
\usepackage{booktabs}
\usepackage{caption}

\theoremstyle{definition}

\newtheorem{example}{Example}

\newtheorem{definition}{Definition}

\theoremstyle{plain}
\newtheorem{theorem}{Theorem}

\newtheorem{proposition}{Proposition}
\newtheorem{corollary}{Corollary}
\theoremstyle{remark}
\newtheorem{remark}{Remark}

\usepackage{cases}

\def\laweq{\buildrel \mathrm{d} \over =}

\theoremstyle{definition}

\def\N{\mathbb{N}}

\def\p{\mathbb{P}}
\def\E{\mathbb{E}}

\def\R{\mathbb{R}}

\def\d{\mathrm{d}}
\DeclareMathOperator*{\esssup}{ess\text{-}sup}

\usepackage[onehalfspacing]{setspace}

\usepackage{bm}
\usepackage{tikz-qtree}
\usepackage{tikz}

\def\id{\mathds{1}}

\pgfplotsset{compat=1.18}

\title{Risk--insurance parity}

\author{Benjamin C\^ot\'e\thanks{Department of Statistics and Actuarial Science, University of Waterloo, Canada.   \texttt{b3cote@uwaterloo.ca}}\and  Ruodu Wang\thanks{Department of Statistics and Actuarial Science, University of Waterloo, Canada.   \texttt{wang@uwaterloo.ca}} \and  Qinyu Wu\thanks{School of Management, University of Science and Technology of China. \texttt{wu051555@mail.ustc.edu.cn}} }

 \date{December 10, 2025 \\
        This version: \today}

\begin{document}
	\maketitle
	\begin{abstract}
	Risk aversion and insurance are two prominent and interconnected concepts in economics and finance.
    To explore their fundamental connection, 
    we introduce risk--insurance parity, which associates various classes of insurance contracts with different notions of risk aversion. 
   We show that the classic notions---both weak  and strong---of risk aversion can be characterized by propensity to different classes of insurance contracts,  generalizing   recent results on propensity to full, proportional, and deductible-limit contracts in the literature. We examine in depth classes of insurance indemnity functions that correspond to weak and strong risk aversion. We show that one proportional-indemnity contract is enough to obtain equivalence with strong risk aversion.
   Risk--insurance parity allows us to define two new notions of risk aversion, between weak and strong, characterized by insurance propensity to deductible-only and limit-only contracts, respectively. 

	\end{abstract}

 \noindent\rule{\textwidth}{0.5pt}

\section{Introduction}
Weak risk aversion, introduced by \cite{A63} and \cite{P64}, and strong risk aversion, introduced by \cite{RS70}, are classical concepts at the cornerstone of the theory of choice under risk. While both concepts coincide in the commonly employed expected utility (EU) decision theory framework, they are not always equivalent, for example in the dual utility model of \cite{Y87} and the more general rank-dependent utility (RDU) model; see \cite{C95} for a summary.
  Properly understanding risk aversion and its implications is of high relevance to economics, as it is the compulsion propelling an agent facing risk to buy insurance.  A risk-averse agent will indeed, by definition, seek to reduce their risk and prefer actions contributing to this. Such actions include the purchase of insurance, which by design transfers (a part of) uncertainty to the insurer. Risk aversion thus underlies the very purpose of insurance. 

To be precise, however, both definitions of risk aversion suppose individuals are willing to buy insurance \textit{when sold at the pure premium}.  
In practice, insurance is generally not sold at an actuarially fair premium. To infer risk aversion of an agent from behaviors in an insurance context is therefore not as straightforward as it may seem. At a premium above the actuarially fair benchmark, a risk-averse individual may refuse to buy insurance. Conversely, accepting such a premium for a particular risk does not, by itself, establish risk aversion in general. The decision to purchase insurance reflects both the individual's risk preferences and the premium charged. 


Recently, \cite{MMWW25} suggested an alternative approach of examining whether an individual prefers to participate in an insurance contract over any other equally distributed contracts; they qualify such an attitude as   \textit{insurance propensity}. By holding the payoff distribution fixed, this approach focuses on the contract’s dependence on the underlying risk without requiring the premium to be actuarially fair. One of their results indicates that weak risk aversion is equivalent to insurance propensity to full-indemnity contracts, tested for every risk and every premium. 
Note the equivalence: Insurance propensity is no longer simply a behavioral consequence of risk aversion, it entirely characterizes it. 
\cite{MMWW25} also showed that strong risk aversion 
is equivalent to insurance propensity to all proportional-indemnity contracts, when tested for all risks, all coinsurance proportions and all premiums. It is also equivalent to insurance propensity to all deductible-limit contracts, when tested for all risks, all deductibles, all limits and all premiums. 
These results raise three important questions:
\begin{enumerate}
    \item The existing characterizations require insurance propensity to hold for entire families of contracts, with premiums and indemnity parameters varying over their admissible ranges. Such variation complicates the empirical elicitation of risk preferences, essentially making them infeasible. Can these contract sets be reduced substantially and still yield equivalence between insurance propensity and risk aversion? In particular, is it possible to reduce them to even just one contract?
     

    \item Proportional-indemnity and deductible-limit contracts have different structures, yet insurance propensity to either family characterizes strong risk aversion. What common feature explains this equivalence? More generally, what properties of a contract set are necessary and sufficient for insurance propensity to be equivalent to weak or strong risk aversion?

    \item Deductible-only and limit-only contracts are also widely used in insurance practice, yet insurance propensity to either family does not coincide with weak or strong risk aversion. Can insurance propensity to either be employed to describe other, perhaps novel, notions of risk aversion?
\end{enumerate} 

To study these questions, we introduce \textit{risk--insurance parity}. This property of a contract set indicates that insurance propensity to all contracts in the set is equivalent to a type of risk aversion. 
Risk--insurance parity is the underlying concept that unifies the above three questions. 
We 
call it \textit{first} risk--insurance parity when insurance propensity to the contract set is equivalent to weak risk aversion, and \textit{second} risk--insurance parity when it is equivalent to strong risk aversion.\footnote{We prefer to adopt this terminology over \textit{weak} and \textit{strong} risk--insurance parity because one is not a weaker concept of the other: neither implies the other since risk--insurance parity requires precisely the equivalence to the corresponding concept of risk aversion.} We work with a general framework for risk preferences, with no assumptions on the preferences other than a form of continuity and choice under risk.
Risk--insurance parity is   defined to hold for \emph{all} continuous risk preferences.

We answer the first 
question by a surprisingly simple ``yes'' in Section~\ref{sect:ONE}.   
One of our main results therein is that any proportional-indemnity contract individually satisfies second risk--insurance parity. 
This is an important reduction compared to the set of \textit{all} proportional contracts in \citet[Theorem~3]{MMWW25}. 
This reduction brings the connection between strong risk aversion and insurance propensity closer to the situation an agent may face in reality, where they are only offered one proportion and  one premium by the insurer for an insurance contract, and this equips the theory of \cite{MMWW25} with a much simpler testable ground.

The second question is addressed 
in Section~\ref{sec:first-second}, which  contains characterization results on first and second risk--insurance parity. 
 We show that full-indemnity contracts are the only types of contracts satisfying first risk--insurance parity. We also provide characterizations of contract sets satisfying second risk--insurance parity when they comprise only one contract, have strictly increasing indemnities, or spawn a dense set of premiums.

The third question is addressed in Section~\ref{sect:Others}. 
Beyond providing another perspective on classical concepts of risk aversion, risk--insurance parity offers a way to define new notions of risk aversion, by letting them arise from insurance propensity to specific classes of contracts. We thereby proceed for the classes of deductible-only and limit-only contracts. Each leads us, through risk--insurance parity, to define new concepts of risk aversion: respectively, right-handle risk aversion and left-handle risk aversion. These new concepts fall strictly between weak and strong risk aversion, and they are quite customary compared to these classical notions: Both have an interpretation in terms of mean-preserving spreads and, in the dual utility model, can be characterized by a criterion on the shape of the weighting function. 



The paper is organized as follows. In Section~\ref{sect:prelim}, we collect some preliminaries on the risk preferences.  Section~\ref{sect:contracts} introduces insurance contracts and insurance propensity. Sections~\ref{sect:ONE}, \ref{sec:first-second} and \ref{sect:Others} contain our main results mentioned above.
In these sections, we consider the space of simple random variables for a better intuition, and  
Section~\ref{sect:L1} extends all main results to $L^p$ spaces to allow for broader applications. 
Section~\ref{sec:concl} concludes. 
 

\section{Preliminaries on the theory of choice under risk}
\label{sect:prelim}

We fix an atomless probability  space $(\Omega, \mathcal{F}, \mathbb{P})$. Let $\mathcal{S}$ be the set of all simple random variables.\footnote{A random variable is simple if it only takes a finite number of values. } We will refer to $X\in\mathcal{S}$ as a risk random variable, where positive values represent actual losses and  negative values represent surpluses. Its law is given by the pushforward measure $\p\circ X^{-1}$. Let $F_X$ be the distribution function of $X$, defined by $F_X(x)=P(X\le x)$ for $x\in\R$. Let $F_X^{-1}$ denote the left-continuous  inverse of $F_X$, also known as the quantile function.
We use $\E[X]$ to represent the expectation of $X$ under $\p$. For simplicity, we regard any real number $x\in\R$ as a constant (or degenerate) random variable. Let $\id_A$ denote the indicator function of a set $A$, that is, $\id_A(\omega) = 1$ if $\omega \in A$ and $\id_A(\omega) = 0$ otherwise. For any $x, y \in \R$, we use the notation: $x \wedge y = \min\{x, y\}$, $x \vee y = \max\{x, y\}$, and $x_+ = \max\{x, 0\}$.  For any $k\in\mathbb{N}$, let $[k] = \{1,\ldots,k\}$, and $[0] = \varnothing$. 
Throughout, terms such as  ``increasing'' and ``convex'' are always understood in the non-strict sense.

The decision maker's preferences are represented by a binary relation $\succsim$ on $\mathcal S$, with 
$\simeq$ denoting its symmetric part.

\begin{definition}
A binary relation $\succsim$ on $\mathcal{S}$  is \emph{law invariant} if for all $X,Y\in \mathcal{S}$, the equality in distribution $X\laweq Y$ implies indifference $X\simeq  Y$.
  A \emph{risk preference} is a law-invariant and transitive binary relation on $\mathcal{S}$.  
\end{definition}

Law invariance 
is also called probabilistic sophistication (\citealp{MS92}). In what follows, $\succsim$ will always be a risk preference.
We say a risk preference is \emph{continuous} if it is continuous under bounded convergence. Specifically, a sequence $\{X_n: n\in\mathbb{N}\} \subseteq \mathcal{S}$ is said to converge to $X \in \mathcal{S}$ in terms of bounded convergence, denoted by $X_n \stackrel{\rm B}{\to} X$, if $X_n \to X$ pointwise and $\sup_{n \in \mathbb{N}} \|X_n\| < \infty$, where $\|\cdot\|$ denotes the $L^\infty$ norm. A preference $\succsim$ is continuous if for any such sequences $\{X_n:n\in\N\}$ and $\{Y_n:n\in\N\}$ in $\mathcal{S}$ with $X_n \stackrel{\rm B}{\to} X$ and $Y_n \stackrel{\rm B}{\to} Y$, it holds that
  \begin{equation*}
  X_n \succsim  Y_n  \quad \text{for all }n\in\mathbb{N} \quad \implies \quad X\succsim Y. 
  \end{equation*}

To select bounded convergence as mode of convergence for the continuity of risk preferences is upheld by \cite{CM95}, notably.
 For complete risk preferences,\footnote{A risk preference is complete if for any $X,Y\in \mathcal S$, at least one of $X\succsim Y$ and $Y\succsim X$ holds.}  continuity ensures the existence of a continuous representing function (\citealp[Proposition 3.4]{H89}) and is equivalent to the sometimes preferred, generally weaker, assumption of hemicontinuity (\citealp[Lemma 3]{W54}). Continuity plays a part in being able to characterize risk aversion wholly from the insurance propensity to a restricted set of contracts. 
 In what follows, we always assume risk preferences are continuous; we continue to state it in formal results.   

Besides continuity, we do not make any further assumption on risk preferences; and thus, we remain in arguably a very general framework for choice under risk.\footnote{ Our framework of risk preferences includes rank-dependent utility (\citealp{Q82, Y87}),
betweenness preferences (\citealp{D86, C89, G91}), multiplier preferences (\citealp{HS08}) and cautious expected utility (\citealp{CDO15}).} In particular, we do not assume risk preferences are complete, meaning a decision-maker could be indecisive between two risks. This is the case, e.g., for mean-variance preferences (\citealp{M52}) and other incomplete decision models (e.g., \citealp{DMO04}).

Next, we present two classical notions of risk aversion. Below is the definition of weak risk aversion, introduced by \cite{A63} and \cite{P64}.  
\begin{definition}
\label{def:wra}
    A risk preference $\succsim$ exhibits \textit{weak risk aversion} if $\mathbb{E}[X]\succsim X$ for all $X\in\mathcal{S}$. 
\end{definition}

 Weak risk aversion means preferring sure payoffs over uncertain payoffs with the same expected value. 
One may however argue that riskiness is more complex a notion than a certainty--uncertainty dichotomy, and thus risk aversion should account for such nuances. 
This leads to the concept of strong risk aversion, as put forth in the seminal paper by \cite{RS70}. 
The definition of strong risk aversion is given in Definition~\ref{def:sra} below; it requires to introduce the convex order first. 

The convex order $\leq_{\rm cx}$ is a partial order on $\mathcal{S}$ defined by stating that 
$X\leq_{\rm cx}Y$ when $\mathbb{E}[\psi(X)]\leq \mathbb{E}[\psi(Y)]$ holds for every convex function $\psi:\mathbb{R}\mapsto\mathbb{R}$.
One may consult \citet[Chapter~1]{MS02} and \citet[Chapter~3]{SS07} for an overview of the convex order and its properties. 
Note that we treat  $X$ and $Y$ as losses instead of surpluses, and hence a larger element in convex order means higher risk. 
Up to a sign flip,
	 convex order (for losses) is economically equivalent to concave order (for gains) in the sense that  $X \le_{\rm cx} Y$ 
	if and only if $-Y \le_{\rm cv} -X$ (and equivalently $Y \leq_{\rm cv} X$), 
	where $\le_{\rm cv}$ is defined analogously to $\leq_{\rm cx}$ by changing convex functions to concave functions. Concave order is closely related to the concept of mean-preserving spread introduced by  \cite{RS70}, which is used to define strong risk aversion.

\begin{definition}
    \label{def:sra}
    A risk preference exhibits \textit{strong risk aversion} if for all $X, Y \in \mathcal{S}$,  $X\leq_{\rm cx}Y$ implies $X\succsim Y$. 
\end{definition}

As signified by their names, strong risk aversion implies weak risk aversion; this is evident from their respective definitions, since $\mathbb{E}[X]\leq_{\rm cx} X$ for any $X\in \mathcal{S}$.

\section{Contracts, insurance propensity, and risk--insurance parity}
\label{sect:contracts}

We present the insurance setting upon which our results on insurance propensity are designed. We introduce indemnity functions, premiums and premium principles, building blocks of (insurance) contracts. We formally define insurance propensity using these notions. A formal definition of risk--insurance parity closes the section, along with some additional comments on the setting assumptions. 

\subsection{Indemnity functions}
\label{sect:indemnity}
An \emph{indemnity function} is a nonconstant increasing function $I:\mathbb{R}\to\mathbb{R}$ and represents the amount paid to the agent after risks have concretized. 
The random variable $I(X)$ is understood as $\omega\mapsto I (X(\omega))$. 
The elementary requirement for $I$ to be increasing ensures that indemnities are not reduced as losses aggravate.
\begin{remark}
Although typical insurance losses are nonnegative, we allow $I $ to be defined on and take values in $\R$, which should be seen as a natural extension from the practical range. Similarly, all our parameters, including the premiums, the deductibles, and the limits, are allowed to take values in $\R$.  
The reason of this extension is that concepts of risk aversion 
are defined for both positive and negative random variables, and therefore we cannot only restrict to positive losses; this is consistent with the setting of
 \cite{MMWW25}.  Conceptually, such an extension does not create troubles, as we work with bounded random variables, and we can shift negative values to positive values by adding a constant background wealth. To be precise, when the risk random variable $X$ takes negative values, we should interpret it as $X=Z-w$, where $Z\ge 0$ is the actual loss, and $w$ is the constant surplus of the decision maker.
 The   payment $I(X)$ should then be interpreted also as an indemnity on $Z$ with surplus $w$.  
 This interpretation fits well with our study, as all notions of risk aversion that we consider are  invariant under translation, allowing us to make constant shifts. 
 \end{remark}

Some examples of indemnity functions are presented below. 
\begin{enumerate}
\item \textbf{Full.} When $I$ is the identity function, $I(x) = x$ for $x\in\mathbb{R}$, we call $I$  a \textit{full-indemnity function}; it pays the totality of a loss.  
\item \textbf{Proportional.} For  $\alpha\in (0,1)$, 
we call $I:x\mapsto\alpha x$ a \textit{proportional-indemnity function} and $\alpha$ is called the coinsurance proportion. 
Proportional indemnities are common for health insurance policies, which often let insureds pay a fraction of their expenses to deter overuse of health services.  
\item \textbf{Deductible.}  The function $I: x\mapsto (x-d)_{+}$ for $x\in\mathbb{R}$, with $d\in\mathbb{R}$, is called a \textit{deductible-indemnity function}, with $d$ being the deductible. Deductibles are common for insurance policies as they allow insurers to screen out small claims and avoid the operational cost of processing these. 
\item \textbf{Limit.} The function $I:x \mapsto x\wedge \lambda$, where $\lambda\in\mathbb{R}$, is referred to as a \textit{limit-indemnity function}, with $\lambda$ being the limit. Limits are included in insurance policies to protect the insurer from having to pay for very large claims which could compromise their solvency. 
\item \textbf{Deductible-limit.} Some indemnity schedules combine both deductible and limit: $I: x \mapsto (x-d)_+\wedge \lambda$, where $d\in\mathbb{R}$ and $\lambda>0$; we say it is a \textit{deductible-limit indemnity function}. 
Property insurance policies usually include both a deductible and a limit.  
\end{enumerate}

\subsection{Premium principles and contract sets }

A \emph{premium} is a real number $\pi\in\mathbb{R}$ representing the amount paid by the agent to participate in a contract. Premiums are usually nonnegative because they are a loss for the agent; however, as for indemnities, we do not make this restriction. (An agent could \textit{be paid} to participate in a contract.) 

A \emph{premium principle} is a functional $\rho: \mathcal{S}\mapsto \mathbb{R}$ that is law invariant.\footnote{We say a functional $\rho:\mathcal S\to\R$ is law invariant if $\rho(X)=\rho(Y)$ whenever $X\laweq Y$.}  It serves the same role as a premium; only, it accounts for the distribution of the risk underlying the contract. 
Law invariance is a widely undisputed axiom for insurance pricing (\citealp{WYP97}), and all common insurance premium principles are law invariant (see e.g., the textbook \citealp{D17}). 
Although the literature abounds in recommendations for additional, desirable properties of premium principles (e.g., \citealp{D90b}, \citealp{WYP97}), we do not require any  for our analysis besides law invariance. We do not even suppose monotonicity. 
The absence of such requirements comes from the fact that we will be comparing identically distributed risks, and premium principles will only serve to bridge elements of $\mathcal{S}$ to one another.


A \emph{contract (function)}  is a mapping $C:\mathcal{S}\to \mathcal{S}$ with one of the following two forms, both common in insurance practice:
\begin{enumerate}
    \item[(a)]   $C(X) = \pi - I(X)$, with   an indemnity $I$ and  a premium $\pi$; such a contract will be called a  \emph{standard contract};
    \item[(b)]   $C(X) = \rho(I(X)) - I(X)$, with an indemnity $I$ and premium principle $\rho$; such a contract will be called a \emph{$\rho$-priced contract}. 
\end{enumerate}

Standard contracts are a special case of $\rho$-priced contracts where $\rho$ is simply a constant. Note that $\rho$-priced contracts use both the distribution of $X$ via the functional $\rho\circ I$ and its outcome via the function $I$. For standard contracts, let any additive constant in the indemnity function be absorbed into the premium, and results referring to contract types or their parameters will suppose this throughout.


A \emph{contract set}, which will be commonly denoted by 
 $\mathcal{C}$, is  a set of   contracts. It is called a \textit{standard contract set} if 
 all contracts in $\mathcal{C}$  are standard, and it is called a \textit{$\rho$-priced contract set} if
all contracts  in $\mathcal{C}$ are $\rho$-priced for the same $\rho$.   

\subsection{Insurance propensity}

We say that a contract is an \textit{insurance} contract if it is applied to the random variable that is precisely the risk faced by the agent who purchases it. 
 
\begin{definition}
\label{def:inspropensity}
    A risk preference $\succsim$ exhibits \textit{insurance propensity} to a contract set $\mathcal{C}$ when, for all $X,Y \in \mathcal{S}$ and $C\in  \mathcal{C}$,  
    $$
    X\laweq Y ~\implies~ 
    X+C(X) \succsim X + C(Y).
    $$
\end{definition}


Insurance propensity means that, with consideration of the risk they are currently facing, an agent always prefers participating in an insurance contract over any other identically distributed contract. 

\begin{remark}
\cite{MMWW25} defined insurance propensity by: 
    for any $X,Z\in L^{\infty}$ and $ C\in\mathcal{C}$,  
       $$  C(X) \laweq Z 
        ~\implies~ X+ C(X)  \succsim X+Z .
 $$ Here, $Z$ takes the place of $C(Y)$. 
 One can check that the two formulations are equivalent. 
 To see this, it suffices to verify that 
$Z$ can be written as $C(Y)$ for some $Y\laweq X$.  
 This follows from
\citet[Corollary 6.11]{K02} under the assumption that the probability space is atomless. 
We prefer to express insurance propensity as in Definition~\ref{def:inspropensity}   because it adds clarity in our statements and proofs.
\end{remark}

Let us give a simple example to put insurance propensity in context. 
\begin{example}
\label{ex:house}
 A firm is exposed to a loss $X\in \mathcal{S}$. 
 An analyst for the firm compares investment opportunities: either they can hedge their loss with a coverage perfectly aligned with their loss, or choose a payment linked to another index, imperfectly aligned. For the sake of comparison, suppose that both indices have the same payout distribution and cost.    
We observe the actions of the firm when presented with different loss-coverage contracts:
\begin{itemize}
    \item $C_{\rm A} : X \mapsto 750 - X$, they choose to hedge their loss using the perfectly aligned coverage; 
    \item $C_{\rm B} : X \mapsto 1000 - 0.5 X$, they choose the imperfectly-aligned index-linked payment;
    \item $C_{\rm C} : X \mapsto 60 - X \wedge 500$, they choose the imperfectly-aligned index-linked payment. 
\end{itemize}
The analyst states that these would be the firm's choices regardless of the distribution of $X$. With this information, we conclude that the firm is insurance propense to $\mathcal{C} = \{C_{\rm A}\}$, but not to $\mathcal{C}^{\prime} = \{C_{\rm A},  C_{\rm B},  C_{\rm C}\}$. In this scenario, insurance propensity is interpreted as aversion to basis risk. 
\end{example}

Insurance propensity should not be understood as whether an agent wants to participate in a contract or not. It really is whether the agent prefers an insurance contract over identically distributed contracts. 
In Example~\ref{ex:house}, the firm is not given the choice of abstaining from any loss-coverage opportunities. This decision rather depends on the degree of their risk aversion and whether the hedge premium mirrors it adequately. 
Insurance propensity, by itself, only establishes comparisons between insurance and non-insurance contracts. 
In the scenario of Example~\ref{ex:house}, the firm gives no information on their preference between two imperfect indices. However, as we will discuss in the next subsection, if the set of contracts is rich enough, it is possible to deduce, from insurance propensity, some preferences between two non-insurance positions. 



\subsection{Risk--insurance parity}
\label{sect:parity}

The main results of this paper concern when insurance propensity to a contract set $\mathcal{C}$ is equivalent to some concept of risk aversion. We call \textit{risk--insurance parity} the property of a contract set to satisfy such equivalences. 

\begin{definition} 
\label{def:1st-r-i-p}
We say that a contract set $\mathcal C$ satisfies \emph{first (resp.~second) risk--insurance parity}
when for all continuous risk preferences, insurance propensity to $\mathcal{C}$ is equivalent to weak (resp.~strong) risk aversion.
\end{definition}  

First and second risk--insurance parity do not imply each other; indeed, they are conflicting because weak and strong risk aversion are not equivalent.

Risk--insurance parity can be explained by the language of dependence. 
 Two random variables $X,Y\in \mathcal{S}$ are \emph{comonotonic} when $(X(\omega_1)-X(\omega_2))(Y(\omega_1)-Y(\omega_2))\geq 0$ for all $\omega_1,\omega_2\in \Omega$.  Two random variables $X,Y\in \mathcal{S}$ are \emph{counter-monotonic} when $X$ and $-Y$ are comonotonic; see \cite{PW15} for a summary on these concepts.
Since indemnity functions are increasing,  $X$ and $C(X)$ are always counter-monotonic for any contract $C$. 
Insurance propensity can therefore be understood as preference for counter-monotonicity in a specific set of situations. Risk--insurance parity indicates when this set of situations is exhaustive enough for this preference for counter-monotonicity to be equivalent to weak or strong risk-aversion. 

Eventually, what influences preferences is the dependence relation between $X$ and $C(Y)$, and how it compares with the counter-monotonic dependence relation between  $X$ and $C(X)$. 
This approach notably circumvents the considerations on whether the premium in $C$ is actuarially fair, which is a highly non-trivial question; see, e.g., \cite{B80} and \cite{B85}. Such considerations rather relate to the marginal distribution of $C(Y)$, which is the same as $C(X)$. Because the suitability of a premium principle usually depends on the context of its application (for instance, see \citealp{C88}, \citealp{W00} and \citealp{GL11}),  to circumvent these considerations thus also contributes to let our analysis be applicable to a very general context of choice under risk.

\begin{remark}
    All results of this paper also hold when trading ``risk aversion" for ``risk propensity (also called risk seeking)" and ``insurance propensity" for ``insurance aversion" in their statements. 
    Thus, every type of risk--insurance parity also describes the equivalence between risk propensity and insurance aversion.
    \end{remark}

\section{One contract to rule them all}
\label{sect:ONE}

As stated in the Introduction, existing results  from \cite{MMWW25} all require some form of exhaustiveness from standard contract sets to satisfy first or second risk--insurance parity. For example, the set of all full-indemnity contracts for all premiums $\pi\in\mathbb{R}$ satisfies first risk--insurance parity \citep[Theorem 1]{MMWW25}, and the set of all proportional-indemnity contracts for all premiums $\pi\in\mathbb{R}$ and all coinsurance proportions $\alpha\in(0,1)$ satisfies second risk--insurance parity \citep[Theorem 2]{MMWW25}. 
In the case of first risk--insurance parity, no substantial reduction of such a standard contract set is possible (see Theorem~\ref{th:1st-r-i-p-standard} in Section~\ref{sec:first-second} below). However, it is possible to satisfy second risk--insurance parity with just one contract, as shown in the following theorem. 

\begin{theorem}
\label{th:ONEprop}
    A standard contract set made of any one proportional-indemnity contract alone satisfies second risk--insurance parity. 
\end{theorem}

In particular, by choosing a contract with premium $\pi =0$, Theorem~\ref{th:ONEprop} implies that, for any (possibly non-complete) continuous risk preference $\succsim$ and $\alpha\in(0,1)$, 
\begin{equation}
\label{eq:SRA-Nice}
    X\succsim Y ~ \text{whenever} ~ X\leq_{\rm cx} Y
 ~\iff~ 
    (1-\alpha) X  \succsim X - \alpha Y  ~ \text{whenever} ~ X\laweq Y.  
    \end{equation}
This is, to our knowledge, a novel characterization of strong risk aversion, and it remarkably only requires equality in distribution as auxiliary relation on random variables.  
The following example illustrates how one may apply Theorem~\ref{th:ONEprop} to infer strong risk aversion. 
\begin{example}
We consider again the situation in Example~\ref{ex:house} where a firm has to choose between a perfectly- and an imperfectly-aligned indices. 
After the analyst indicates that it is preferable to choose the perfectly-aligned index for the contract $C_{\rm B} : X\mapsto 1000 - 0.5X$, they specify that the distribution of $X$ had no incidence on their decision. 
This indication means that the firm is insurance propense to $\{C_{\rm B}\}$. 
 By Theorem~\ref{th:ONEprop}, $\{C_{\rm B}\}$ satisfies second risk--insurance parity, and therefore the firm is strongly risk averse. 
\end{example}

Another method to reduce the number of contracts needed to satisfy risk--insurance parity is to resort to $\rho$-priced contract sets. For first risk--insurance parity, this permits the sought reduction, impossible with standard contracts, to only one contract. We indeed need just one $\rho$-priced full-indemnity contract to satisfy first risk--insurance parity, with only a weak property required from the premium principle $\rho$. 

We introduce the \emph{range property} of a premium principle $\rho$:
\begin{equation}
     \{\rho(X+a) : a\in\mathbb{R}\} \mbox{ is dense in $\R$ for each }X\in \mathcal{S}. \label{eq:propertyA}  \tag{RP} 
\end{equation} 
This property is very weak. In particular,  \eqref{eq:propertyA} is weaker than cash-additivity,\footnote{A premium principle $\rho$ is cash-additive if $\rho(X+b) =\rho(X)+b$  for all $b\in\mathbb{R}$ and all $X\in \mathcal S$.} a property often argued to be desirable for premium principles, see e.g.~\citet[Section 2.2.3]{DDGK05}. All monetary risk measures introduced by \cite{FS16} satisfy cash additivity and hence \eqref{eq:propertyA};  indeed they satisfy the stronger property that $\{\rho(X+a):a\in \R\}=\R$ for each $X\in \mathcal{S}$.

\begin{theorem}
\label{th:1st-r-i-p-priced}
    A $\rho$-priced contract set $\mathcal{C}$ satisfies first risk--insurance parity if it comprises solely a full-indemnity contract and the premium principle $\rho$  satisfies the range property \eqref{eq:propertyA}. 
\end{theorem}

The range property~\eqref{eq:propertyA} allows for the vastness of considered random variables to be translated to the premiums and hence removes the exhaustiveness of premiums otherwise required for standard contract sets. The situation in Theorem~\ref{th:1st-r-i-p-priced} is thus closer to the reality of a policyholder since, for any risk they insure, they will only be proposed one premium by the company (often a function of the distribution of risk --- i.e., the contract is $\rho$-priced). The policyholder makes their decision in consideration of that premium and not of all premiums which could have been proposed.  
The following example shows an application of Theorem~\ref{th:1st-r-i-p-priced} to infer weak risk aversion. 

\begin{example}
    A firm is given the choice between hedging a loss with an index perfectly aligned to this loss or another index, imperfectly aligned. Assume for the sake of comparison that all indices are identically distributed to the loss. The hedging premium is the same for each opportunity, and corresponds to the median of the loss.  
    If no matter the distribution of the loss, the firm always prefers the perfectly aligned index, we deduce that it is weakly risk averse, since $\mathcal{C} = \{C: X \mapsto  F_{X}^{-1}(0.5)-X\}$ satisfies first risk--insurance parity by Theorem~\ref{th:1st-r-i-p-priced}. Interestingly, a similar deduction would not have been possible if participating in the contracts was free in all situations (no hedging cost): the contract set $\mathcal{C} = \{C: X\mapsto -X \}$ does not satisfy first risk--insurance parity by Theorem~\ref{th:1st-r-i-p-standard} (Section~\ref{sec:first-second}).   
    \end{example}

The $\rho$-priced approach similarly allows to construct singleton contract sets satisfying second risk--insurance parity. 
    In this vein,  we introduce the \emph{matching property}:  for some $\gamma>0$,
\begin{equation}
    \{ b\in \R: \rho(X-b) = \gamma b\} ~~\mbox{is nonempty for each $X\in \mathcal{S}$.} \tag{MP}\label{eq:propertyB}
\end{equation}
We call \eqref{eq:propertyB} the matching property because typically, $\rho(X-b)$ is decreasing in $b$, and $\gamma b$ is increasing in $b$, so we naturally expect that they match at some  point $b\in \R$ under mild regularity conditions. For instance, if $b\mapsto \rho(X-b)$ is decreasing and continuous (satisfied by   all common premium principles), then \eqref{eq:propertyB} holds, and this includes the case where $\rho$ is a constant (as in a standard contract).  
Properties \eqref{eq:propertyA} and \eqref{eq:propertyB}
may  not  be familiar, but they are implied by standard properties in the literature.  
The next result  gives  an example of a $\rho$-priced singleton contract set satisfying second risk--insurance parity.

\begin{proposition}
\label{th:2nd-r-i-p-priced-contracts-prop}
   Let $\alpha\in(0,1)$, and
     $\rho$ be a premium principle satisfying the matching property \eqref{eq:propertyB} for $\gamma = (1-\alpha)/\alpha$. The $\rho$-priced contract set $\mathcal{C}$ comprising solely $C:X\mapsto\rho(\alpha X) - \alpha X$ satisfies second risk--insurance parity.
\end{proposition}

The singleton contract sets in Theorem~\ref{th:ONEprop} and Proposition~\ref{th:2nd-r-i-p-priced-contracts-prop} are practical in that they allow to establish strong risk aversion after verifying insurance propensity only for a specific contract. Meanwhile, the larger contract sets provided by Proposition~\ref{prop:bigger-contract-sets} below extend this insurance propensity to a greater set of situations, and are hence practical for applications thereafter.

\begin{proposition}
\label{prop:bigger-contract-sets}
    For two contract sets $\mathcal{C}_1$ and $\mathcal{C}_2$, if $\mathcal{C}_1$ satisfies second risk--insurance parity and $\mathcal{C}_1\subseteq \mathcal{C}_2$, then $\mathcal{C}_2$ also satisfies second risk--insurance parity.
\end{proposition}

By
 Proposition~\ref{prop:bigger-contract-sets},  the set of all contracts is the largest set satisfying second risk--insurance parity.

\section{Characterization results on risk--insurance parity}
\label{sec:first-second}

This section contains further results on first and second risk--insurance parity.
Section~\ref{sect:ONE} gave examples of contract sets satisfying each; here, we examine what specific properties make a contract set satisfy them. 

\subsection{A full characterization of first risk--insurance parity}
\label{sect:weak-premium}


We characterize first risk--insurance parity for standard contract sets in the following theorem. 
\begin{theorem}
\label{th:1st-r-i-p-standard}
    A standard contract set satisfies first risk--insurance parity if and only if it comprises  full-indemnity contracts whose premiums form a dense set in $\mathbb{R}$.
\end{theorem}

The main reason why one needs to consider all $\pi$ in a $\mathbb{R}$-dense set is that full-indemnity contracts offer no flexibility in relying on the full space of $\mathcal{S}$: all randomness is removed after one participates in an insurance contract. Hence, to recover the weak risk aversion behavior, one needs to fish out the specific contract whose premium is the expectation of the risk underlying the contract. 


 A characterization result for $\rho$-priced contract sets (i.e., adding an ``only if'' statement in Theorem~\ref{th:1st-r-i-p-priced}) is more intricate. As we  establish in the proof of Theorem~\ref{th:1st-r-i-p-standard}, weak risk aversion may only effectively induce insurance propensity to full-indemnity contracts, no matter whether standard or $\rho$-priced. The issue rather comes from requiring \eqref{eq:propertyA} to hold for all $X\in \mathcal{S}$, since we only need it to hold on a set $\mathcal{Z}\subseteq \mathcal{S}$  satisfying 
 \begin{equation*}
 \text{for every }X \text{ in a dense subset of } \mathcal{S}\text{, there exist }Z,W\in\mathcal{Z}\text{ such that }X\laweq Z - W + \mathbb{E}[X]\text{ and }Z\laweq W,
 \end{equation*}
 where density is understood with respect to bounded convergence.
 We remark the complexity of characterizing such a set by noting that $\mathcal{Z}$ is strictly smaller than $\mathcal{S}$:  degenerate random variables, for example, are not required to be in $\mathcal{Z}$ since any $Z,W \in \mathcal{Z}$, $Z\laweq W$, can yield $Z-W = 0$ by taking $Z$ and $W$ to be comonotonic. 

\subsection{A full characterization of second risk--insurance parity with one contract}

We first present 
a full characterization
of second risk--insurance parity 
for contract sets comprising only one single contract with continuous and strictly increasing indemnity. Define the range of retained loss for a standard contract $C: X\mapsto \pi - I(X)$ as  
$$J_C = \{x + \pi - I(x): x\in \R \},$$
which corresponds to the range of all possible losses for an insured individual, for all possible risks. 
 
\begin{theorem}\label{th:solo}
    A standard contract set comprising only one contract $C$  with continuous strictly increasing indemnity  satisfies second risk--insurance parity if and only if $ J_C = \mathbb{R}$. 
\end{theorem}

A notable type of indemnity that is continuous, strictly increasing and satisfies that $J_C = \mathbb{R}$, as in the conditions of Theorem~\ref{th:solo}, is the proportional indemnity. Hence, an immediate consequence of Theorem~\ref{th:solo} is Theorem~\ref{th:ONEprop} stated in Section~\ref{sect:ONE}.

The proof of Theorem~\ref{th:solo}  requires a technical result on strong risk aversion,
which may be of independent interest. 
To interpret this result, we first present the classic notion of mean-preserving spreads and contractions, well known to be closely related to strong risk aversion \citep{RS70}.  
 
\begin{definition}
\label{def:mps-def}
    We say that a random variable $Y\in \mathcal{S}$ is a \emph{mean-preserving spread} of $X\in \mathcal{S}$ when
\begin{equation}
\label{eq:mps-def}
Y \laweq X - \delta_1 \id_{A_1} + \delta_2 \id_{A_2},
\end{equation}
for some $\delta_1, \delta_2 > 0$ and $A_1,A_2\in\mathcal{F}$ such that $ \delta_1\p(A_1)=\delta_2 \p(A_2)$ and $X(\omega_1) \leq X(\omega_2)$ for all $\omega_1 \in A_1$ and $\omega_2\in A_2$. Conversely, we say that $X$ is a \emph{mean-preserving contraction} of $Y$, and this can be written as 
\begin{equation}
\label{eq:mpc-def}
    X \laweq Y + \delta_1 \id_{A_1} - \delta_2 \id_{A_2},
\end{equation}
for some $\delta_1,\delta_2>0$, with $\delta_1 + \delta_2 \leq \inf Y(A_2) - \sup Y(A_1)$, with $A_1,A_2\in\mathcal{F}$ such that $\delta_1 \p(A_1) = \delta_2 \p(A_2)$ and $\sup Y(A_1) < \inf Y(A_2)$.  A mean preserving spread/contraction is \textit{simple} if $\delta_1 = \delta_2$. 
\end{definition}

Strong risk aversion can be elicited from the existence of preferences over simple mean preserving contractions, as shown in the following proposition.
 
\begin{proposition}
\label{lem:mpc-sra}
For any continuous risk preference $\succsim$, the following are equivalent:
\begin{enumerate}[label={\rm (\roman*)}, ref={\rm (\roman*)}]
\item\label{item:mpc-sra-1}  $\succsim$ is strongly risk averse; 
\item\label{item:mpc-sra-2} for every $Y\in\mathcal{S}$, and equiprobable events $A_1, A_2\in\mathcal{F}$ such that $Y=a_1$ on $A_1$ and $Y=a_2$ on $A_2$, with $a_1< a_2$, it holds that $Y + \delta\id_{A_1} - \delta\id_{A_2} \succsim Y$ for all $\delta\in(0,(a_2-a_1)/2]$; 
\item\label{item:mpc-sra-3} for every $Y\in\mathcal{S}$, and equiprobable events $A_1, A_2\in\mathcal{F}$ such that $Y=a_1$ on $A_1$ and $Y=a_2$ on $A_2$, with $a_1< a_2$, it holds that $Y + \delta^*\id_{A_1} - \delta^*\id_{A_2} \succsim Y$ for at least one $\delta^* \in (0, (a_2-a_1)/2]$.
\end{enumerate}
\end{proposition}

 This proposition is useful for establishing second risk--insurance parity; it is also, in itself, an interesting result on strong risk aversion. The equivalence between items \ref{item:mpc-sra-1}--\ref{item:mpc-sra-2} is not surprising given the classical result \citep[Lemma~1]{RS70} that preference for mean-preserving contractions entirely describes strong risk aversion.  The equivalence between items \ref{item:mpc-sra-2}--\ref{item:mpc-sra-3} indicates that the preference for one mean-preserving contraction for each random variable is the same as preference for all mean-preserving contractions. 
 Surprisingly enough, a similar statement does not hold for mean-preserving spreads; a counterexample in this regard is provided in Appendix~\ref{counter:mpc-sra}. For exhaustiveness of comparisons to be induced, 
 the fixed risk in the comparisons must indeed be the more dispersed one.

\subsection{Second risk--insurance parity for multiple contracts 
}

We next remove the restriction of the contract set being a singleton.
 For strictly increasing and continuous indemnities, we derive a characterization result for second risk--insurance parity with a simple criterion on the range of retained loss $J_C$.

\label{sect:strong-premium}
\begin{theorem}
\label{th:full-strict}
Let $\mathcal{C}$ be a set of standard contracts 
where each indemnity function is continuous
and strictly increasing. Then $\mathcal{C}$ satisfies
second risk--insurance parity if and only if  every
 bounded interval $A\subseteq \R$ is contained in
$ 
 J_C 
$ for some $C\in \mathcal C$.

\end{theorem}

If $\mathcal{C}$ is finite, the condition in Theorem~\ref{th:full-strict} is equivalent
to the existence of a contract $C\in\mathcal{C}$ satisfying
 $ J_C=\mathbb{R}.
 $
 The result's noteworthiness lies especially in its ``only if'' part: it indicates that combining finitely many contracts with one-side bounded $J_C$, even if their union is $\mathbb{R}$, will not yield second risk--insurance parity. One really needs some notable individual contracts with unbounded $J_C = \mathbb{R}$, or tending to $\mathbb{R}$, for the contract set to satisfy risk--insurance parity when indemnities are strictly increasing and continuous.

Deductible-limit contracts do not satisfy the conditions from Theorem~\ref{th:solo} 
because their indemnities are not strictly increasing. The next proposition gives sufficient conditions for a contract set with (non-strictly) increasing indemnities to satisfy second risk--insurance parity, by relying on a large set of premiums.


\begin{proposition}\label{prop:unbounded}
Consider the contract set $\mathcal{C}$ comprising contracts $C_{\pi}: X \mapsto \pi - I(X)$ for every $\pi\in\Pi$, where $\Pi$ is a dense subset of $\mathbb{R}$, and all contracts share the same continuous indemnity function. If $J_C = \mathbb{R}$ for at least one $C\in\mathcal{C}$, then $\mathcal{C}$ satisfies second risk--insurance parity. 
\end{proposition}
 
In particular, in Proposition~\ref{prop:unbounded}, $\mathcal{C}$ satisfies second
risk--insurance parity whenever the indemnity function is bounded.
A direct consequence of this observation is the following corollary. 

\begin{corollary}
\label{th:ded-lim}
 For fixed deductible $d\in\mathbb{R}$ and limit $\lambda>0$, the standard contract set of deductible-limit contracts for all premiums $\pi\in \Pi$ satisfies second risk--insurance parity, where $\Pi$ is a dense subset of $\mathbb{R}$. 
\end{corollary}

The contract set in Corollary~\ref{th:ded-lim}, although infinite, is nonetheless a reduction from the contract set derived in \citet[Theorem 2]{MMWW25}, where all deductibles and limits needed to be considered. As explicited in the proof of Proposition~\ref{prop:unbounded} by the requirement on the range $J_C$, the shape of the retention function $x\mapsto x + \pi -I(x)$ is determinant in whether insurance propensity can be translated to preference for mean-preserving contractions. For deductible-limit indemnity, the retention plateaus from $x=d$ to $x=d+\lambda$, and insurance propensity means indifference to additional loss on that portion; this is what elicits preference for the insurance under deductible-limit contract, and thus preference for a mean-preserving contraction provided the plateau coincides with its location. The dense set of premiums is required to perform a translation of the retention function so that the plateau indeed coincides. In comparison, proportional indemnity leaves the retention function with a slope strictly less than 1, indicating a lesser increase in retained loss compared to the true loss. This holds simultaneously for all $\mathbb{R}$; and thus no translation via the premium is necessary.

The contract set of all deductible-indemnity contracts encompassing all premiums and all deductibles does not satisfy second risk--insurance parity. This may be surprising considering the ensuing implication that insurance propensity to one proportional-indemnity contract alone induces a stronger notion of risk aversion than the contract set of all deductible contracts. Deductible contracts force the retention function to be constant on a right half-line, which then occults the behavior on that segment. Similarly, limit contracts force the retention function to be constant on a left half-line. The type of risk--insurance parity induced by each is examined in the next section.

\section{Risk--insurance parity with deductible and limit contracts}
\label{sect:Others}

 Whereas in the previous sections risk--insurance parity served to highlight the connection between insurance propensity and existing notions of risk aversion, we now employ risk--insurance parity as a way of defining new notions of risk aversion, by letting this behavior arise from insurance propensity to a given  contract set that is practically relevant.

In Corollary~\ref{th:ded-lim} and in \cite{MMWW25}, the deductibles and limits are always present together for contracts inducing strong risk aversion. This leads us naturally to  the question of what types of risk aversion would be equivalent to insurance propensity to deductible-indemnity contracts and limit-indemnity contracts, respectively. 
These types of contracts  and their approximations  are common in health and property insurance practice. 
It turns out that the induced notions of risk aversion are strictly between weak and strong, and both are, to the best of our knowledge, new to the literature. 
Each new notion has an intuitive formulation based on mean-preserving spreads and a concise and explicit characterization within the dual utility model  (as all considered notions of risk aversion are equivalent for the expected utility model).   

We first define the new notions of risk aversion through mean-preserving spreads, and then discuss risk--insurance parity pertaining to each. 
This  reverse-route organization helps to develop intuition about the notions of risk aversion before providing the keystone results characterizing their equivalence to insurance propensity to deductible and limit contract sets. 
The specific design of these notions of risk aversion was, in fact, driven by the  risk--insurance parity results.

\subsection{Right- and left-handle mean-preserving spreads}

Before introducing the new notions of risk aversion, we first define two types of mean-preserving spreads. 

\begin{definition}
\label{def:leftright-handle}
    Consider a mean preserving spread as in \eqref{eq:mps-def}. 
We say that it is a \emph{left-handle mean-preserving spread} if $X(\omega_1) = \min X$ for all $\omega_1\in A_1$. 
It is a \emph{right-handle mean-preserving spread} if  $X(\omega_2) = \max X$ for all $\omega_2\in A_2$. 
\end{definition}

The new mean-preserving spreads in Definition \ref{def:leftright-handle} are called left-handle and right-handle as one may imagine a handle at each of the extremities of the distribution: Pulling on the left handle to produce a mean-preserving spread may only decrease the value of outcomes that are minimal for $X$; pulling on the right handle may only increase the values of outcomes that are maximal for $X$. 
An illustration of these types of mean-preserving spreads is provided in Figure~\ref{fig:mean-preserving-spreads} and makes the analogy more clear.

\begin{figure}[t!]
\centering
 \includegraphics[width=\textwidth]{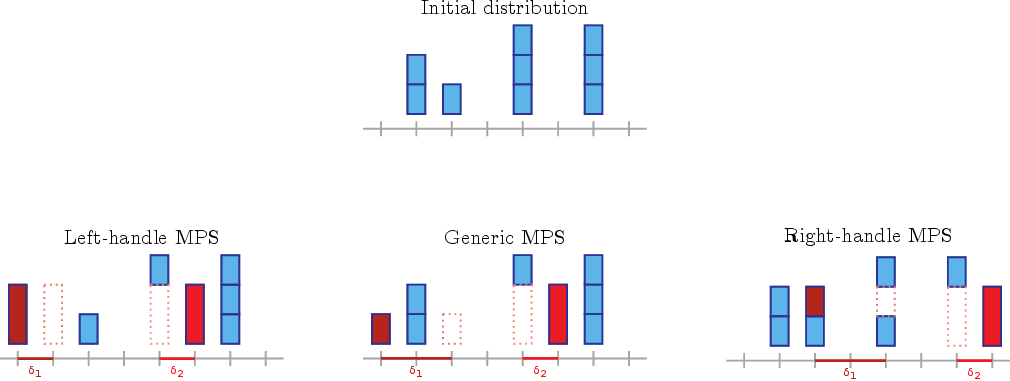}
 \captionsetup{font=small}
 \caption{\small (Center top.) The initial distribution on which the mean-preserving spreads (MPS) are performed. (Left.) A left-handle mean-preserving spread.  (Center bottom.) A mean-preserving spread that is neither left- nor right-handle. (Right.) A right-handle mean-preserving spread. Each block represents an atom of probability mass.
}\label{fig:mean-preserving-spreads}
\end{figure}

\begin{remark}
\label{rem:EMPS}
  Every mean-preserving spread is expressible as the result of a sequence of simple ones. Thus, due to the transitivity and continuity of risk preferences, there is no discrepancy between aversion to mean-preserving spreads and to simple mean-preserving spreads.
    However, the possibility of $\delta_1$ and $\delta_2$ taking distinct values, and thus $A_1$ and $A_2$ not being equiprobable, has an incidence for left- or right-handle mean-preserving spreads, as it truly adds a layer of flexibility in their respective definitions. A left-handle (resp. right-handle) mean-preserving spread cannot be turned into many simple ones if $\delta_1<\delta_2$ (resp. $\delta_1 >\delta_2$). 
\end{remark}
\begin{remark}
    A mean preserving spread that is both left- and right-handle is strong in the sense of \cite{MO85}. However, not all strong mean-preserving spreads are left- or right-handle ones because they allow for all the transported mass to be taken from the center part of the distribution, not necessarily its extremities.   
\end{remark}

As we already stated above, for  $X,Y\in\mathcal{S}$, we have that $X\le_{\rm cx}Y$ if and only if $Y$ can be obtained from $X$ through a finite number of mean-preserving spreads. The finite number can be zero,  in which case $X\laweq Y$. In a similar vein, we define stochastic orders pertaining to left- or right-handle mean-preserving spreads as follows: For $X,Y\in\mathcal{S}$,
\begin{itemize}
\item We say that $X$ is dominated by $Y$ in left-handle convex order, denoted by $X\le_{\rm lhcx} Y$ if $Y$ can be obtained from $X$ through  a finite number of left-handle mean-preserving spreads.

\item We say that $X$ is dominated by $Y$ in right-handle convex order, denoted by $X\le_{\rm rhcx} Y$ if $Y$ can be obtained from $X$ through  a finite number of right-handle mean-preserving spreads.
\end{itemize}
It is straightforward to verify that both orders satisfy reflexivity, transitivity and antisymmetry. 
In the following result, we provide a characterization of left- and right-handle convex orders in terms of the compared random variables' distributions.


\begin{proposition}\label{prop:charac-handles}
For $X,Y\in\mathcal S$, let $(X^{\rm co},Y^{\rm co})$ be comonotonic with $X^{\rm co}\laweq X$ and $Y^{\rm co} \laweq Y$. Suppose $X\not \laweq Y$. The following two statements hold:
\begin{enumerate}[label={\rm (\roman*)}, ref={\rm (\roman*)}]
\item\label{item:charac-left-handle} $X\le_{\rm lhcx} Y$ if and only if $\E[X]=\E[Y]$, $\min X > \min Y$, and 
$ X^{\rm co} \le Y^{\rm co} $ a.s.~on the event $\{X^{\rm co}>\min X\}$. 

\item\label{item:charac-right-handle} $X\le_{\rm rhcx}Y$ if and only if $\E[X]=\E[Y]$, $\max X < \max Y$, and $X^{\rm co} \ge Y^{\rm co}$ a.s.~on the event $\{X^{\rm co}<\max X\}$.
\end{enumerate}
\end{proposition}


Directly from Proposition~\ref{prop:charac-handles},
we have that $X\le_{\rm lhcx}Y$ or $X\le_{\rm rhcx}Y$ implies that the distribution functions of 
 $X$ and $Y$ possess the single-crossing property. This is consistent with the fact that the left- and right-handle convex orders are strictly stronger than the convex order.



The following corollary is a direct consequence of Proposition~\ref{prop:charac-handles}, which states that two loss random variables can be compared under the left-handle convex order if and only if their surplus counterparts can be compared under the right-handle convex order. 

\begin{corollary}\label{co:r&lhcx}
For $X,Y\in\mathcal S$, $X\le_{\rm lhcx} Y$ if and only if $-X\le_{\rm rhcx} -Y$.
\end{corollary}

We give an additional closedness property of left- and right-handle convex orders, which will be convenient in the proofs of some results. 

\begin{proposition}\label{prop:handlecontinuity}
The partial orders $\le_{\rm rhcx}$ and $\le_{\rm lhcx}$ are closed with respect to bounded convergence. 
\end{proposition}

\subsection{Right- and left-handle risk aversion and risk--insurance parity}
\label{sect:right-handle}

We introduce below right-handle risk aversion. To the best of our knowledge, this notion is new to the literature. 

\begin{definition}
\label{def:right-handle-r-a}
    A risk preference $\succsim$ exhibits \textit{right-handle risk aversion} if, for every $X,Y\in \mathcal{S}$ where $X\leq_{\rm rhcx} Y$, it holds that $X\succsim Y$. 
\end{definition}

Right-handle risk aversion represents fear of the worst-case scenario: a right-handle risk averse individual will prefer to reduce the variability of their loss when it also means it reduces their worst possible outcome. 
This leaves some room for risk-loving behaviors. For example, the individual may like to gamble small amounts of money, while still preferring to insure their house. 

A direct consequence of 
Proposition \ref{prop:charac-handles} is that $\E[X]\le_{\rm rhcx} X$ for all $X\in\mathcal S$. Hence, right-handle risk aversion is stronger than weak risk aversion. 
 On the other hand, because right-handle risk aversion is defined using a specific case of mean-preserving spreads, it is weaker than strong risk aversion. In Section~\ref{sect:yaari}, we will show that these inclusions hold strictly by characterizing these notions of risk aversion in dual utility (\citealp{Y87}).

\begin{remark}
    Under the expected utility (EU) framework, weak risk aversion and strong risk aversion are equivalent in that they are both characterized by a concave utility function. Since right-handle risk aversion lies conceptually between weak and strong risk aversion, it is also 
    characterized by concavity of utility functions in the EU framework. 
\end{remark}


In line with the main considerations of this article, we now examine under which set of contracts insurance propensity and right-handle risk aversion are equivalent.  

\begin{definition}
    A contract set $\mathcal{C}$ satisfies \textit{right-handle risk--insurance parity} when, for all continuous risk preferences, insurance propensity to $\mathcal{C}$ is equivalent to right-handle risk aversion.  
\end{definition}

 For a deductible-indemnity standard contract with premium $\pi$ and deductible $d$, let $\theta = \pi+d$. We will refer to $\theta$ as the \textit{price parameter} of the deductible contract.

\begin{theorem}
\label{th:right-handle-r-i-p-if}
A standard (resp.~$\rho$-priced, with any premium principle $\rho$)  contract set $\mathcal{C}$ satisfies right-handle risk--insurance parity if it comprises precisely 
\begin{enumerate}
\item[(a)]  deductible-indemnity contracts whose price parameters (resp.~deductibles) are dense  in $\mathbb{R}$; 
\item[(b)] any, possibly none, full-indemnity contracts.
\end{enumerate}
 \end{theorem}

Theorem~\ref{th:right-handle-r-i-p-if} shows that deductible-indemnity contracts are indeed the type of contracts inducing right-handle risk aversion through risk--insurance parity. Insurance contracts of this type provide a coverage for risks greater than the deductible, but the portion of the risk below that deductible remains uncovered. This observation provides some high-level insight to explain the asymmetry in right-handle risk aversion.




We are now interested in whether other types of contract sets may satisfy right-handle risk--insurance parity. In other terms, we aim to characterize right-handle risk--insurance parity analogously to the results on first and second risk--insurance parity in the previous section.
To proceed, we will need to restrict our attention to Lipschitz-continuous (LC) indemnity functions, that is, $I$ satisfying $|I(x) - I(y)| \leq |x-y|$ for all $x,y\in\mathbb{R}$. While we did not require it thus far, this assumption is common for indemnity functions to deter sabotage; see, e.g., \citet[Section~4]{HMS83}. 
Proportional, deductible and limit indemnity functions satisfy LC.  An \emph{LC contract set} is a contract set such that the indemnity functions of its elements are LC.

 It turns out that deductible-only contract sets are the only LC contract sets that satisfy right-handle risk--insurance parity.
 This is made explicit in the following theorem. 
\begin{proposition}
    \label{th:right-handle-r-i-p-shape}
    An LC contract set satisfying right-handle risk--insurance parity exclusively contains deductible-indemnity or full-indemnity contracts. 
\end{proposition}

In Proposition~\ref{th:right-handle-r-i-p-shape}, we consider, in the case of $\rho$-priced contract sets, possible additive constants in the indemnities because we did not assume them to be absorbed into the premiums. Hence, for example, $I(x) = x+a$ would be considered a valid full indemnity. 
We now give the following characterization of right-handle risk--insurance parity for LC standard contract sets. 

\begin{proposition}
\label{prop:right-handle-r-i-p-double-equiv}
A standard  LC contract set $\mathcal{C}$ satisfies right-handle risk--insurance parity if and only if it comprises precisely 
\begin{enumerate}
\item[(a)]  deductible-indemnity contracts whose price parameters  are dense  in $\mathbb{R}$; 
\item[(b)] any, possibly none, full-indemnity contracts.    
\end{enumerate}
\end{proposition}

 Disregarding the LC assumption, one may construct noncontinuous contracts that, when added to the deductible-only contract set, still preserve right-handle risk–insurance parity; see Example~\ref{ex:non-LC} in Appendix~\ref{sec:ex-non-LC}. However, including such irregular contracts would significantly complicate the analysis and obscure the interpretability of the characterization result. In particular, the necessity part of Proposition~\ref{prop:right-handle-r-i-p-double-equiv} would become more cumbersome to formulate, weakening both the elegance and the clarity of the theoretical insight.

While Theorem~\ref{th:right-handle-r-i-p-if} and Proposition~\ref{th:right-handle-r-i-p-shape} apply to both standard and $\rho$-priced contract sets, obtaining a result similar to Proposition~\ref{prop:right-handle-r-i-p-double-equiv} for $\rho$-priced contract sets is arduously feasible. The exhaustiveness of price parameters is intricate to translate to $\rho$-price contract sets, since it will vary with $\rho$ and depend on the set of indemnity random variables required when establishing that insurance propensity implies right-handle risk aversion.

We next define left-handle risk aversion, the symmetric counterpart to right-handle risk aversion. 

\begin{definition}
\label{def:left-handle-r-a}
    A risk preference $\succsim$ exhibits \textit{left-handle risk aversion} if, for every $X,Y\in \mathcal{S}$ where $X\le_{\rm lhcx}Y$, it holds that $X\succsim Y$. 
\end{definition}

A left-handle-risk-averse decision-maker prefers to reduce their best-case-scenario outcome if it means less variability in their loss. An example of such would be an individual who prefers to buy a cheaper lottery ticket, with a smaller jackpot, over a pricier lottery ticket with a bigger jackpot, where both options have the same expected gain. In terms of insurance-related risks, left-handle risk aversion is perhaps less natural as individuals buying insurance usually worry less about their minimal loss than the catastrophic case of total loss. 
From its symmetry with right-handle risk aversion, left-handle risk aversion is strictly between weak  and strong risk aversion. 

Next, we are interested in the set of contracts under which insurance propensity and left-handle risk aversion coincide.

\begin{definition}
    A contract set $\mathcal{C}$ satisfies \textit{left-handle risk--insurance parity} when, for all continuous risk preferences, insurance propensity to $\mathcal{C}$ is equivalent to left-handle risk aversion.  
\end{definition}

Let us introduce an alternate formulation of limit-indemnity contracts: 
consider the $\rho$-priced contract $C_{\zeta} : X\mapsto \rho(X\wedge\zeta - \zeta) - (X\wedge\zeta - \zeta)$ with $\zeta\in\mathbb{R}$. This adjustment represents taking a portion of the premium and considering it as a basis for the indemnity; this is to allow $\zeta$ to exert shifts in the distribution, as a deductible  does for deductible-indemnity contracts. We will refer to this alternate formulation as a \textit{leveled-limit} contract and $\zeta$ as its \textit{level parameter}. 

\begin{theorem}
\label{th:left-handle-r-i-p-if}
A standard (resp.~$\rho$-priced, with any premium principle $\rho$) contract set $\mathcal{C}$ satisfies left-handle risk--insurance parity if it comprises precisely 
\begin{enumerate}
\item[(a)] limit-indemnity (resp.~leveled-limit) contracts whose premiums (resp.~level parameters) are dense in $\mathbb{R}$;
\item[(b)] any, possibly none,  full-indemnity contracts. 
\end{enumerate}
 \end{theorem}

Limit-only insurance contracts provide full coverage below the limit, but no additional payout above; this connects to the asymmetric behavior in left-handle risk aversion. 
LC contract sets satisfying left-handle risk--insurance parity exclusively comprise limit-only contracts, as indicated in the next theorem.

\begin{proposition}
\label{th:left-handle-r-i-p-shape}
    A LC contract set satisfying left-handle risk--insurance parity exclusively contains limit-indemnity or full-indemnity contracts.     
\end{proposition}

As for Proposition~\ref{th:right-handle-r-i-p-shape}, we allow for additive constants in the indemnity functions of Proposition~\ref{th:left-handle-r-i-p-shape} for $\rho$-priced contract sets; for standard contracts, they are already absorbed into the premium. 
We now give the following characterization of left-handle risk--insurance parity for standard LC contract sets. 
\begin{proposition}
\label{th:left-handle-r-i-p-double-equiv}
A standard LC contract set $\mathcal{C}$ satisfies left-handle risk--insurance parity if and only if it comprises precisely 
\begin{enumerate}
\item[(a)] limit-indemnity contracts whose premiums are dense in $\mathbb{R}$;
\item[(b)] any, possibly none,  full-indemnity contracts. 
\end{enumerate}
 \end{proposition}

We speak of \emph{dual-handle risk aversion} when both left-handle and right-handle risk aversion hold.
A contract set $\mathcal{C}$ satisfies \textit{dual-handle risk--insurance parity} if insurance propensity to $\mathcal{C}$ is equivalent to dual-handle risk aversion.  
Combining Theorems~\ref{th:right-handle-r-i-p-if} and~\ref{th:left-handle-r-i-p-if}, we directly obtain the following result.

\begin{corollary}
\label{cor:dual-handle}
    A standard (resp. $\rho$-priced, with any premium principle $\rho$) contract set $\mathcal{C}$ satisfies dual-handle risk--insurance parity if it comprises precisely  
    \begin{enumerate}
       \item[(a)]  deductible-indemnity contracts whose price parameters (resp.~deductibles) are dense  in $\mathbb{R}$;
        \item[(b)]  limit-indemnity (resp. leveled-limit) contracts whose premiums (resp.~level parameters) are dense in $\mathbb{R}$;
        \item[(c)] any, possibly none, full-indemnity contracts. 
    \end{enumerate}
\end{corollary}




In the dual utility model, dual-handle risk aversion has an interpretation in terms of diversification incentives; see the next subsection.

\subsection{Right- and left-handle risk aversion in the dual utility model}
\label{sect:yaari}

Let us characterize right-, left- and dual-handle risk aversion in the dual utility model of \cite{Y87}.
In this model, risk preferences are represented as follows:
For $X,Y\in\mathcal{S}$,
\begin{equation}
    X \succsim Y \quad \iff \quad U_h(-X) \geq U_h(-Y)\label{eq:yaari}
\end{equation}
with $U_h : \mathcal{S}\mapsto \mathbb{R}$ being defined as 
\begin{equation*}
U_h(Z)=\int_{-\infty}^0 \left(h(\p(Z> t))-1\right)\d t+\int_0^\infty h(\p(Z> t))\d t, \quad Z\in\mathcal{S},
\end{equation*}
where $h:[0,1]\to[0,1]$ is a weighting function that is increasing and satisfies $h(0)=1-h(1)=0$.
A weighting function $h$ is \emph{star-shaped} at $a\in [0,1]$ when $h(\lambda a + (1-\lambda)x) \leq \lambda h(a) + (1-\lambda) h(x)$ holds true for all $\lambda, x\in[0,1]$. 
\begin{theorem}
\label{prop:yaari-right-handle}
    In the dual utility model, an agent is 
    \begin{enumerate}[label={\rm (\roman*)}, ref={\rm (\roman*)}]
        \item \label{item:star-1} right-handle risk averse if and only if the weighting function is star-shaped at $1$; 
        \item \label{item:star-0} left-handle risk averse if and only if the weighting function is star-shaped at $0$; 
        \item \label{item:star-0-1} dual-handle risk averse if and only if the weighting function is  star-shaped  at both $0$ and $1$. 
    \end{enumerate}
\end{theorem}


We provide an illustration of a weighting function that is both star-shaped at $0$ and $1$ in Figure~\ref{fig:star-shaped}.
Convexity is a stronger property than star-shapedness at $0$ or $1$: it corresponds to star-shapedness at every point. In the dual utility model, an agent is strongly risk averse if and only if their weighting function is convex (see e.g., \citealp[Theorem~2]{Y87} and \citealp[Corollary~3]{SZ08}). An agent is weakly risk averse if and only if their weighting function is dominated by the identity function (see \citealp[Corollary 1]{CC94}), and this property of the weighting function is weaker than star-shapedness at $0$ or $1$.

\begin{figure}[t!]
\centering
 \includegraphics[width=5.6cm]{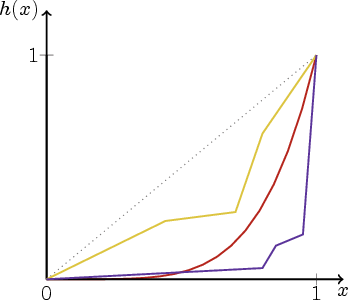}
 \captionsetup{font=small}
 \caption{\small A weighting function that is both star-shaped at 0 and at 1 but not convex (in dark blue), a convex weighting function (in red), and a weighting function that is dominated by the identity, but neither convex nor star-shaped at 0 or 1 (in yellow).
}\label{fig:star-shaped}
\end{figure}

\begin{remark} Star-shapedness of the weighting function at~1 (resp.~0) corresponds to the aversion to left-monotone (resp.~right-monotone) increase in risk; see~\citet[Theorems 4--5]{CCM04}. Therefore, by Theorem~\ref{prop:yaari-right-handle}, in the dual utility model, the concepts of right-handle (resp.~left-handle) and aversion to left-monotone (resp.~right-monotone) increases in risk coincide. (Mind the mismatch between right and left here.)  It is straightforward to see, for example by juxtaposing Proposition~\ref{prop:charac-handles} and the definitions of left-/right-monotone increases in risk in \cite{CCM04}, that in general every right-handle mean-preserving spread is left-monotone but not the converse; thus, beyond the dual utility model, right- (resp.~left-) handle risk aversion is weaker than aversion to left- (resp.~right-) monotone increases in risk. 
\end{remark}

It is straightforward to check that star-shapedness at both $0$ and $1$ implies 
\emph{dual superadditivity}, that is, $h(x)+h(y)\le h(x+y)$  
and $ \tilde h(x)+ \tilde h(y)\ge \tilde h(x+y)$ for $x,y$ with $x+y\in [0,1]$, where $\tilde h(t)=1-h(1-t)$ for $t\in [0,1]$.
Dual superadditivity  is connected to a specific form of diversification. 
A preference relation  $\succsim$ exhibits diversification (\citealp{D89}) if 
\begin{equation}
    \label{eq:convexity}
 X\simeq Y  \implies 
    \lambda X+(1-\lambda) Y\succsim X  \mbox{~for~}\lambda \in [0,1], 
\end{equation}
meaning that combining equally preferable payoffs is desirable. 
\citet[Theorem 2]{GRW25} showed that dual subadditivity is equivalent to a concept called pseudo-convexity introduced by \cite{ACV21}, which is in turn equivalent to diversification restricted to counter-monotonic random variables (\citealp[Proposition~8]{PWW25}). Note that indifference in \eqref{eq:convexity} is satisfied whenever $X\laweq Y$ by law invariance.  Therefore, we have the following implications for the dual utility model. 
\begin{enumerate}[label=(\alph*)]
\item Weak risk aversion implies that \eqref{eq:convexity} holds
when $\lambda X+(1-\lambda) Y$ is a constant.
\item  Dual-handle risk aversion  implies that  
\eqref{eq:convexity} holds when $X$ and $Y$ are counter-monotonic. 
\item Strong risk aversion  implies that  
\eqref{eq:convexity} holds for all $X,Y.$
\end{enumerate}
Since the condition that $\lambda X+(1-\lambda)Y$ is a constant is stronger than counter-monotonicity, the above three implications illustrate the relative strengths of the three notions of risk aversion in terms of their diversification incentives in the dual utility model: Weakly risk-averse agents diversify when risks fully vanish after hedging; 
dual-handle risk-averse agents diversify when risks are perfectly negatively dependent; 
strongly risk-averse agents  diversify in all cases.

\section{Extending the space of random variables}
\label{sect:L1}

Most of our previous results were stated for risk preferences defined on $\mathcal{S}$, the set of simple random variables. In this section, we discuss how our results translate to larger sets of random variables.  Let $L^{\infty}$ be the set of all essentially bounded random variables on $(\Omega, \mathcal{F}, \mathbb{P})$ and $L^{p}$ be the set of all random variables with finite $p$th moment, $p\in [1,\infty)$.
In what follows,
for a risk preference defined on $L^p$ for $p\in [1,\infty)$, continuity is with respect to the $L^p$-norm.
For a risk preference defined on $L^\infty$, continuity is with respect to bounded convergence, as in the main part of our paper.

The following proposition indicates that our definition of insurance propensity could have equivalently been formulated on $L^p$, for any $p\in [1,\infty]$. 
Extending the space to $L^p$ is useful for applications where potential losses may be unbounded.   To deal with potentially unbounded random variables,  
we will make the assumption that 
any contract's indemnity function is bounded or Lipschitz continuous. This assumption is made to ensure proper integrability, 
and it is consistent with the similar assumption in Section~\ref{sect:right-handle}.


\begin{proposition}
\label{prop:L0-Linfty}
For a standard contract $C$ and  a continuous risk preference defined on $L^p$, $p\in[1,\infty]$, the following are equivalent:
    \begin{enumerate}[label={\rm (\roman*)}, ref={\rm (\roman*)}]
        \item\label{item:S} $X+C(X) \succsim X + C(Y)$ for all $X,Y \in \mathcal{S}$ where $X\laweq Y$; 
        \item\label{item:L1} $X+C(X) \succsim X + C(Y)$ for all $X,Y \in L^p$ where $X\laweq Y$.
    \end{enumerate}
    For a $\rho$-priced contract, such an equivalence holds if $\rho$ is continuous with respect to the $L^{p}$-norm. 
\end{proposition}

The following result extends the preceding risk--insurance parity
results to $L^p$ spaces.

\begin{proposition}
\label{prop:extension-to-Linfty}
The statements of
Theorems~\ref{th:ONEprop}--\ref{th:left-handle-r-i-p-if},
Propositions~\ref{th:2nd-r-i-p-priced-contracts-prop}--\ref{prop:unbounded},
and Corollaries~\ref{th:ded-lim} and~\ref{cor:dual-handle}
hold for continuous risk preferences defined on $L^p$,
for $p\in[1,\infty]$.
For the results concerning $\rho$-priced contracts,
suppose that $\rho$ is continuous with respect to the $L^p$-norm.
\end{proposition}



It is straightforward to check that the result of Theorem~\ref{prop:yaari-right-handle} also extends to risk preferences defined on $L^p$ for $p\in [1,\infty]$, under the additional assumption that   the distortion function $h$ yields a finite value of $U_h$ on $L^p$.

\begin{remark}In Definitions~\ref{def:right-handle-r-a} and~\ref{def:left-handle-r-a}, right-handle and left-handle risk aversions are defined in terms of   subsets of mean-preserving spreads.  
Although the definitions only directly impose preference conditions for discrete and mixed random variables, this has nontrivial implications on continuous random variables because of  continuity (see e.g., Theorem \ref{prop:yaari-right-handle}, which has restrictions for comparing general distributions). 
This is the same for strong risk aversion, which can equivalently be stated in terms of mean-preserving mass transfers and continuity (\citealp{RS70}).
\end{remark}



\section{Conclusion}
\label{sec:concl} 

Risk--insurance parity identifies the equivalence relation between risk aversion and insurance propensity to a set of contracts. 
With this new concept, we answer three main questions with positive answers. First, weak and strong risk aversion can be characterized by contract sets with only one contract, possibly a $\rho$-priced contract. Most notably, one proportional-indemnity contract is able to   characterize strong risk aversion. 
Second, 
first and second risk--insurance parity are satisfied by many different contract sets, and we obtain full characterization for several specific classes.
Third, risk--insurance parity moreover offers a way to define new notions of risk aversion;
 we designed two new notions of risk aversion: right-handle risk aversion and left-handle risk aversion, corresponding to deductible-only contracts and  limit-only contracts. 
These notions are not equivalent to either weak or strong risk aversion; in the dual utility model, they correspond to star-shapedness of the weighting function at 1 and 0, respectively.

 Finally, we remark that not all notions of risk aversion can be characterized by some contract sets via risk--insurance parity, even if they lie between weak and strong risk aversion.  For instance, we can show that the notion of monotone risk aversion of \cite{Q92} is not equivalent to insurance propensity for any Lipschitz-continuous contract set. A detailed study of this fact is provided in Section~\ref{sect:monotone}.

\subsection*{Acknowledgments}
 Ruodu Wang is supported by the Natural Sciences and Engineering Research Council of Canada (RGPIN-2018-03823, RGPAS-2018-522590) and Canada Research Chairs (CRC-2022-00141).

\appendix 

\newpage

\begin{center}
\Large{\textbf{Appendix}}
\end{center}

This appendix has two sections. The first section contains proofs of all results in the main paper. The second section discusses the relation between monotone risk aversion and risk--insurance parity.

\section{Proofs of all results}

\subsection{Proof of Theorem~\ref{th:ONEprop}}

Theorem~\ref{th:ONEprop} is a special case of Theorem~\ref{th:solo}, whose proof follows below (in Appendix \ref{proof:solo}). 

\subsection{Proof of Theorem~\ref{th:1st-r-i-p-priced}} 
\label{proof:1st-r-i-p-priced}



Suppose that $\mathcal C$ comprises solely a full-indemnity contract and the premium principle $\rho$ satisfies the range property \eqref{eq:propertyA}, i.e., $\mathcal C=\{C_{\rho}: X\mapsto \rho(X)-X\}$. First we argue that weak risk aversion implies insurance propensity to $\mathcal{C}$: 
For any $X,Y\in \mathcal{S}$ such that $X\laweq Y$, we have
\begin{equation*}
    X+C_{\rho}(X) = \rho(X) = \mathbb{E}\left[X + C_{\rho}(Y)\right] \succsim X + C_{\rho}(Y), 
\end{equation*}
where the second equality follows by law invariance of $\rho$ and the preference relation by weak risk aversion. 

Conversely, assume that insurance propensity to $\mathcal{C}$ holds, and we will argue that it implies weak risk aversion. For any $X\in\mathcal S$, there exists a bounded sequence of random variables $\{X_n:n\in\mathbb{N}\}\subseteq \mathcal{S}$, each piecewise constant on finitely many events with equal probability, such that $X_n \stackrel{\rm B}{\to} X$; see \citet[Lemma 5]{MMWW25}. 
By \citet[Lemma 1]{MMWW25}, there exist $Z_n,W_n\in\mathcal S$ such that $X_n\laweq Z_n-W_n+\E[X_n]$ and $Z_n\laweq W_n$. Furthermore, because $\rho$ satisfies \eqref{eq:propertyA}, there exists $\{b_{j,n}: j,n\in\mathbb{N}\}\subseteq \mathbb{R}$ such that $\rho(Z_n+b_{j,n})\to \mathbb{E}[X_n]$ as $j\to\infty$. 
The law invariance of $\rho$ yields $\rho(W_n+b_{j,n})=\rho(Z_n+b_{j,n})\to\E[X_n]$ as $j\to\infty$ because $Z_n+b_{j,n}\laweq W_n+b_{j,n}$.
Therefore,
\begin{align*}
    \rho(Z_n + b_{j,n}) &= (Z_n+b_{j,n}) + C_{\rho}(Z_n+b_{j,n}) \succsim (Z_n+b_{j,n}) + C_{\rho}(W_n+b_{j,n}) = Z_n - W_n + \rho(Z_n + b_{j,n}), 
\end{align*}
where the preference relation follows from insurance propensity to $\mathcal{C}$ and $Z_n+b_{j,n}\laweq W_n+b_{j,n}$. By continuity of the risk preference, as $j\to\infty$, $\mathbb{E}[X_n]\succsim Z_n - W_n + \mathbb{E}[X_n]\laweq X_n$ for every $n\in\mathbb{N}$. 
Hence, $\mathbb{E}[X] \succsim X$ also by continuity of risk preferences. 
\hfill$\square$

\subsection{Proof of Proposition~\ref{th:2nd-r-i-p-priced-contracts-prop}}

For the implication strong risk aversion $\implies$ insurance propensity to $\mathcal{C}$: Because $I$ is increasing, for any $Z\in\mathcal{S}$, random variables $Z$ and $C(Z)$ are countermonotonic. Therefore, for any $W\in\mathcal{S}$ with $W\laweq Z$, we have $Z + C(Z)\leq_{\rm cx} Z + C(W)$ by \citet[Corollary 3.28]{R13}, and thus $Z+ C(Z) \succsim Z + C(W)$ for any strongly-risk-averse continuous risk preference $\succsim$. Note that, by law invariance, $\rho(\alpha Z) = \rho(\alpha W)$ for any $Z,W\in\mathcal{S}$ where $Z\laweq W$. 

For the implication insurance propensity to $\mathcal{C}$ $\implies$ strong risk aversion: Let $\succsim$ be a continuous risk preference that is insurance propense to $\mathcal{C}$. 
Consider any $Y\in\mathcal{S}$, and equiprobable events $A_1,A_2\in\mathcal{F}$ such that $Y=a_1$ on $A_1$ and $Y=a_2$ on $A_2$, where $a_1<a_2$. Define $\delta^*=\alpha(a_2 - a_1)/(1+\alpha)$, and let 
\begin{equation*}
    X = Y + \delta^*\id_{A_1} - \delta^*\id_{A_2}.
\end{equation*} 
By Proposition~\ref{lem:mpc-sra}, it is sufficient to show $X\succsim Y$ to prove that 
the risk preference $\succsim$ is strongly risk averse. 

Let $b\in\mathbb{R}$ satisfy 
\begin{equation*}
    \rho\left(\frac{\alpha X}{1-\alpha} - b\right) = \frac{1-\alpha}{\alpha}b;
\end{equation*}
such a $b$ exists because $\rho$ satisfies the matching property \eqref{eq:propertyB} for $\gamma=(1-\alpha)/\alpha$. Define the random variable $Z = X/(1-\alpha) - b/\alpha$, and note that 
\begin{equation}
    X = (1-\alpha)Z + \rho(\alpha Z) = Z + C(Z). 
\label{eq:rhoalpha-1}
\end{equation}
Define the random variable $W$ by interchanging the two constant values of $Z$ on $A_1,A_2$, and leaving it unchanged elsewhere; thus, $W\laweq Z$. Moreover, the difference between the values of $Z$ on $A_2$ and $A_1$ equals
\begin{equation*}
    \frac{a_2 - a_1 -2\delta*}{1-\alpha} = \frac{a_2 - a_1}{1+\alpha} = \frac{\delta^*}{\alpha}. 
\end{equation*}
Therefore, 
\begin{equation}
    Z+C(W) = X + \alpha(Z-W) = X - \delta^*\id_{A_1} + \delta^*\id_{A_2} = Y.
\label{eq:rhoalpha-2}
\end{equation}
Combining \eqref{eq:rhoalpha-1}, \eqref{eq:rhoalpha-2} and the insurance propensity assumption yields $X\succsim Y$. Hence, $\succsim$ is risk averse. \hfill$\square$

\subsection{Proof of Proposition~\ref{prop:bigger-contract-sets}}

Any continuous risk preference that is insurance propense to $\mathcal{C}_2$ is also insurance propense to $\mathcal{C}_1$ and therefore strongly risk averse, by second risk--insurance parity. For strong risk aversion to imply insurance propensity to $\mathcal{C}_2$, the only requirement is for all contracts in $\mathcal{C}_2$ to have an increasing indemnity (which is required, by definition of indemnity functions); see the argument on this matter in the proof of Proposition~\ref{th:2nd-r-i-p-priced-contracts-prop}. \hfill$\square$

\subsection{Proof of Theorem~\ref{th:1st-r-i-p-standard}}



\label{proof:1st-r-i-p-standard}

For the ``if" statement: 
Let $\mathcal{C}=\{C_\pi : X\mapsto \pi-X\mid X\in\mathcal{S}, \pi \in \Pi \}$ where $\Pi$ is dense in $\mathbb{R}$.
The statement that first risk–insurance parity holds for $\mathcal C$ was proved in \citet[Theorem 1]{MMWW25} for the case $\Pi = \mathbb{R}$, under either of the following settings: the space of all bounded random variables, or the space of random variables with finite moments of all orders.  Our proof is a simple modulation of theirs, accounting for other dense sets by leveraging continuity of risk preferences, and the space $\mathcal{S}$ of simple random variables. First, we show that weak risk aversion implies insurance propensity to $\mathcal{C}$.
For any $X,Y\in \mathcal{S}$ such that $X\laweq Y$, and any $C_\pi\in\mathcal{C}$, we have
\begin{equation*}
    X+C_{\pi}(X) = \pi = \mathbb{E}\left[X + C_{\pi}(Y)\right] \succsim X + C_{\pi}(Y), 
\end{equation*}
where the preference relation follows from weak risk aversion. 

Next, we show that insurance propensity to $\mathcal{C}$ implies weak risk aversion. Let $X\in \mathcal{S}$. There exists a bounded sequence of random variables $\{X_n:n\in\mathbb{N}\}\subseteq \mathcal{S}$, each piecewise constant on finitely many events with equal probability, such that $X_n \stackrel{\rm B}{\to} X$; see \citet[Lemma 5]{MMWW25}. Furthermore, we can assume that $\E[X_n]\in \Pi$ for all $n\in\N$ as $\Pi$ is dense in $\R$.
By \citet[Lemma 1]{MMWW25}, for each $n\in\mathbb{N}$, there exist $Z_n,W_n\in \mathcal{S}$ such that $X_n\laweq Z_n - W_n + \mathbb{E}[X_n]$
and $Z_n \laweq W_n$. Then, 
\begin{align*}
\mathbb{E}[X_n] = Z_n + C_{\mathbb{E}[X_n]}(Z_n) \succsim Z_n + C_{\mathbb{E}[X_n]}(W_n) = Z_n - W_n + \mathbb{E}[X_n] \laweq X_n.
\end{align*}
The law invariance of $\succsim$ implies $\E[X_n]\succsim X_n$ for all $n\in\N$.
Hence, we conclude that $\mathbb{E}[X]\succsim X$ by continuity of the risk preference.

For the ``only if" statement: We show that,  whenever $\mathcal{C}$ does not satisfy the conditions in the statement, meaning $\mathcal{C}$ comprises other types of contract or $\Pi$ is not dense in $\mathbb{R}$, there exist a continuous risk preference such that weak risk aversion and insurance propensity to $\mathcal{C}$ are not equivalent. First, suppose $\mathcal{C}$ includes  a standard contract, say $C^{*}: X\mapsto \pi^* - I^*(X)$, that is not a full indemnity contract. Define the risk preference $\succsim_a$ as
\begin{equation*}
    X \succsim_a Y \iff  X = \mathbb{E}[Y] \text{ or }X\laweq Y.
\end{equation*}
The relation $\succsim_a$ is law-invariant, transitive and weakly risk averse. It is also continuous because bounded convergence implies convergence of expectations. Note that $I^*$ and $x\mapsto x-
I^*(x)$ are both nonconstant by assumption. Consider a random variable $Z$ for which both differ on two equiprobable events, and let $W\in\mathcal{S}$ interchange these values. Then $W\laweq Z$, while $Z+C^*(Z)$ is nonconstant and therefore $Z+C^*(Z)\neq E[Z+C^*(W)]$. Moreover, the interchange made $Z+C^*(W)$ have a strictly larger maximum than $Z+C^*(Z)$, so their distributions differ.  Hence, $Z+C^*(Z)\not\succsim_a Z+C^*(W)$, and $\succsim_a$ is not insurance propense to $\mathcal{C}$. This shows that $\mathcal{C}$ does not satisfy first risk--insurance parity. 

Next, suppose that $\mathcal{C}$ only includes full-indemnity standard contracts $C_\pi : X\mapsto \pi - X$, $\pi\in \Pi$, but $\Pi$ is not dense in $\mathbb{R}$. We will show that there exists a continuous risk preference that is insurance propense to $\mathcal{C}$ but not weakly risk averse. Define $\succsim_b$ by
\begin{equation*}
    X\succsim_b Y \iff  (X = \mathbb{E}[Y] \text{ and }\mathbb{E}[Y]\text{ is in the closure of $\Pi$}) \text{ or }X\laweq Y. 
\end{equation*}
The relation $\succsim_b$ is transitive, law invariant and continuous, because bounded convergence implies convergence of expectations. 
Note that for any $Z\in\mathcal{S}$ and $C_{\pi}\in\mathcal{C}$, $\pi\in\Pi$, we have $Z + C_{\pi}(Z) = \mathbb{E}[Z + C_{\pi}(W)] = \pi$ whenever $W\laweq Z$. Hence, $\succsim_b$ is insurance propense to $\mathcal{C}$; however, by definition, it is not weakly risk averse if $\Pi$ is not dense in $\mathbb{R}$.  
 \hfill$\square$

\subsection{Proof of Theorem~\ref{th:solo}}
\label{proof:solo}

The proof requires Proposition~\ref{lem:mpc-sra}, whose proof follows (in Appendix~\ref{proof:mpc-sra}). Let $C:\mathcal{S}\to \mathcal{S}$ denote the sole contract in the contract set $\mathcal{C}$ described in the statement; that is, $C(X) = \pi_0 - I(X)$ where $\pi_0\in\mathbb{R}$ and $I$ is continuous and strictly increasing. 

For sufficiency: Assume that $J_C=\mathbb{R}$, and we need to show that second risk--insurance parity holds. For the implication strong risk aversion $\implies$ insurance propensity to $\mathcal{C}$:  Because $I$ is increasing, for any $Z\in\mathcal{S}$, random variables $Z$ and $C(Z)$ are countermonotonic. Therefore, for any $W\in\mathcal{S}$ with $W\laweq Z$, we have $Z + C(Z)\leq_{\rm cx} Z + C(W)$ by \citet[Corollary 3.28]{R13}, and thus $Z+ C(Z) \succsim Z + C(W)$ for any strongly-risk-averse continuous risk preference $\succsim$. For the implication insurance propensity to $\mathcal{C}$ $\implies$ strong risk aversion: By Proposition~\ref{lem:mpc-sra}, it suffices to prove that, for any random variable $Y\in\mathcal{S}$, and equiprobable events $A_1,A_2\in\mathcal{F}$ such that $Y=a_1$ on $A_1$ and $Y=a_2$ on $A_2$, $a_2 >a_1$, there exists $\delta^* \in (0, (a_2-a_1)/2]$ for which $Y + \delta^*\id_{A_1} - \delta^*\id_{A_2} \succsim Y$. Define $g:\mathbb{R}\to\mathbb{R}$ by letting $g(x) = x+I(x)$ for every $x\in\mathbb{R}$. Because $I$ is continuous and strictly increasing, so is $g$. Therefore, the inverse $g^{-1}$ is well-defined and also continuous and strictly increasing. Let the function $h: \mathbb{R}\to\mathbb{R}$ be defined as $h(x) = x- I(g^{-1}(g(x) + a_2 - a_1))$. This function is continuous and since $g$, $g^{-1}$ and $I$ are all strictly increasing, it must hold that $\{h(x) : x\in\mathbb{R}\}=\mathbb{R}$ as $\{x-I(x) : x\in\mathbb{R}\}=\mathbb{R}$ by assumption. There thus exists $s\in\mathbb{R}$ such that $h(s) = a_1 - \pi_0$. Write $t_s = g^{-1}(g(s) + a_2 - a_1)$.
Note that 
\begin{equation}
t_s - s + I(t_s) - I(s) = g(t_s) - s - I(s) = g(s) + a_2 - a_1  - s - I(s) = a_2 - a_1. 
\label{eq:ts-s-a1-a2}
\end{equation}
The following identities also hold:
\begin{align}
    &s-I(t_s) + \pi_0 = h(s) + \pi_0 = a_1; \label{eq:proofth4-id1}\\
    &t_s - I(s) + \pi_0 = s - I(t_s) + a_2 - a_1 + \pi_0 = h(s) + \pi_0 + a_2 - a_1 = a_2. \label{eq:proofth4-id2}
 \end{align}
For any $y\in\mathbb{R}$, let $z_y\in\mathbb{R}$ be the value satisfying $y=z_y - I(z_y) + \pi_0$; since $J_C=\mathbb{R}$, such a value exists for every $y\in\mathbb{R}$. 
Define the following random variables: 
\begin{equation}
    Z = \begin{cases}
        s & \text{on }A_1,\\
        t_s & \text{on }A_2,\\
        z_Y & \text{elsewhere,}
    \end{cases}, \quad \text{ and }\quad W = \begin{cases}
        t_s & \text{on }A_1,\\
        s & \text{on }A_2,\\
        z_Y & \text{elsewhere.}
    \end{cases}
    \label{eq:z-and-w-one-contract}
\end{equation}
It holds that $Z\laweq W$ because we assumed $A_1$ and $A_2$ to have the same probability. Note that $t_s > s$, so taking $\delta^* = I(t_s) - I(s)$ is legitimate. (By \eqref{eq:ts-s-a1-a2}, we have $I(t_s) - I(s) < a_2 - a_1$, and a mean-preserving contraction operation for a given $\delta^{\dagger}\in((a_2-a_1)/2, a_2 - a_1)$ will produce the same distribution as one with magnitude $(a_2 - a_1 - \delta^{\dagger})\in (0,(a_2-a_1)/2]$ because $Y$ is fixed on $A_1$ and $A_2$.)  
Then, 
\begin{equation*}
    Y + \delta^* \id_{A_1} - \delta^*\id_{A_2} \laweq Z + C(Z) \succsim Z + C(W) \laweq Y, 
\end{equation*}
where the equalities in distribution hold by the identities \eqref{eq:proofth4-id1} and \eqref{eq:proofth4-id2}, see Table~\ref{tab:th4}, and the preference relation by the assumption of insurance propensity to $\mathcal{C}$. By Proposition~\ref{lem:mpc-sra}, the relation above suffices to establish that $\succsim$ is risk averse.

\begin{table}
\centering
\scriptsize\resizebox{\textwidth}{!}{
\begin{tabular}{llll}
\toprule
 &  $\Omega\backslash(A_1\cup A_2)$ & $A_1$ & $A_2$\\
\midrule
$Y$ & $Y(\omega)$ & $a_1$ & $a_2$\\
$Y + \delta^*\id_{A_1} - \delta^*\id_{A_2}$ &  $Y(\omega)$ & $a_1 + I(t_s) - I(s) ~\;[= s + \pi_0 - I(s)]$ & $a_2 - I(t_s) + I(s) ~\;[= t_s +\pi_0 - I(t_s) ]$ \\
$Z$ & $z_{Y(\omega)}$ & $s$ & $t_s$ \\
$W$ &  $z_{Y(\omega)}$ & $t_s$ & $s$ \\
$Z + C(Z)$ & $z_{Y(\omega)} + \pi_0 - I(z_{Y(\omega)}) ~\;[= Y(\omega)]$ & $s + \pi_0 - I(s)$ & $ t_s +\pi_0 - I(t_s)$\\
$Z + C(W)$ & $z_{Y(\omega)} + \pi_0 - I(z_{Y(\omega)}) ~\;[= Y(\omega)]$ & $s + \pi_0 - I(t_s) ~\;[= a_1]$ & $t_s + \pi_0 - I(s) ~\;[= a_2]$ \\
\bottomrule
\end{tabular}}
\caption{Illustration of the contracts' construction for the proof of Theorem~\ref{th:solo}.   
} 
 \label{tab:th4}   
\end{table}

For necessity: We argue that if $J_C \neq \mathbb{R}$, then $\mathcal{C}$ does not satisfy second risk--insurance parity. To show this, 
we will construct a continuous risk preference that is insurance propense to $\mathcal{C}$ but not strongly risk averse. Because $x\mapsto x-I(x)$ is continuous, $J_C$ must be an interval and, from the assumption that it is not equal to $\mathbb{R}$, it must be bounded on at least one side. This means that there is $m\geq 0$ such that either $J_C\subseteq (-\infty, m)$ or $J_C\subseteq (-m,\infty)$. Let $K = \{X\in\mathcal{S} : X\geq -m\; \text{or} \; X \leq m\}$. 
Define the relation $\succsim_{K}$ by letting, for every $X,Y\in\mathcal{S}$, 
\begin{equation*}
    X\succsim_K Y \quad \iff \quad X\laweq Y \text{ or }(X\in K \; \text{and} \; X\leq_{\rm cx} Y)
\end{equation*}
It is straightforward to verify that $\succsim_K$ is law invariant, transitive and continuous. For any $Z\in\mathcal{S}$, we have $\mathrm{supp}(Z + C(Z))\subseteq J_C$ and therefore $Z+C(Z)\in K$. Also, because $Z$ and $C(Z)$ are countermonotonic, $Z+C(Z)\leq_{\rm cx} Z+C(W)$ by \citet[Corollary 3.28]{R13} for any $W\laweq Z$. Therefore, $Z+C(Z) \succsim Z+C(W)$ whenever $W\laweq Z$, so $\succsim_K$ is insurance propense to $\mathcal{C}$. Meanwhile, $\succsim_{K}$ is not strongly risk averse: to see this, consider $X\in\mathcal{S}$ taking values $\pm (m+1)$ with equal probability and $Y\in\mathcal{S}$ taking values $\pm(m+2)$. We have $X\leq_{cx} Y$, but $X\not\succsim_{K} Y$ because $X\not\in K$. \hfill$\square$

\subsection{Proof of Proposition~\ref{lem:mpc-sra}}
\label{proof:mpc-sra}

For  $X,Y\in\mathcal{S}$, it holds that $X\le_{\rm cx}Y$ if and only if $Y$ can be is the limit of a uniformly bounded sequence obtained from $X$ through  successive simple mean-preserving spreads \citep[Lemma~2]{DS17}. This fact combined with transitivity, law invariance and continuity of $\succsim$ yields the equivalence between \ref{item:mpc-sra-1} and \ref{item:mpc-sra-2}. The implication \ref{item:mpc-sra-2}$\implies$\ref{item:mpc-sra-3} is evident. For the implication \ref{item:mpc-sra-3}$\implies$\ref{item:mpc-sra-2}: For any $Z\in\mathcal{S}$, $\varepsilon>0$, consider the relation
\begin{equation}
    Z + \varepsilon\id_{A_1} - \varepsilon\id_{A_2} \succsim Z,  \label{eq:mpc-epsilon}
\end{equation}
where $A_1$, $A_2$ are equiprobable events. 
Consider $Y\in\mathcal{S}$ such that $Y=a_1$ on $A_1$ and $Y=a_2$ on $A_2$, with $a_1<a_2$. Let $\delta_{\max} = (a_2 - a_1)/2$. We know, by assumption \ref{item:mpc-sra-3}, that \eqref{eq:mpc-epsilon} holds for $Z = Y$ and $\varepsilon = \delta^*$, for a specific $\delta^*\in(0,\delta_{\max}]$. Take any $\delta\in(0,\delta_{\max}]$, and we aim to show that \eqref{eq:mpc-epsilon} holds for $Z=Y$ and $\varepsilon = \delta$. 

First, we argue that \eqref{eq:mpc-epsilon} holds true for $Z=Y$ and $\varepsilon=\delta_{\max}$. Write $c=(a_1+a_2)/2$, and let $W_x$ be the random variable taking value $c-x$ on $A_1$ and $c+x$ on $A_2$, for $x\in[0,\delta_{\max}]$, and $W_x = Y$ elsewhere. (Note that $W_{\delta_{\max}} = Y$).
Define the set $\Xi = \{x\in[0,\delta_{\max}] : W_x \succsim  Y\}$. This set is nonempty: it namely comprises $\delta_{\max} -  \delta^*$ by assumption and $\delta_{\max}$ by law invariance. It is closed: for a converging sequence $\{x_n : n\in\mathbb{N}\}\subseteq \Xi$, we have $W_{\lim_{n\to\infty }x_n} \succsim Y $ by the continuity of $\succsim$, and thus $ \lim_{n\to\infty} x_n  \in \Xi$. 
Because of this, $\Xi$ contains a minimal element, say $x_{\min}$. Suppose that $x_{\min}>0$; we will show that this leads to a contradiction. If $\delta_{\max} - \delta^* = 0$, this is direct (since $(\delta_{\max} - \delta^*)\in\Xi$); suppose now that $\delta_{\max} - \delta^*>0$. By the definition of $\Xi$,  it would hold that $W_{x_{\min}}\succsim Y$ and, because $x_{\min} \leq \delta_{\max} - \delta^* < \delta_{\max}$ (since $(\delta_{\max} - \delta^*) \in \Xi$), $W_{x_{\min}}\not\laweq Y$. Now, by the assumption that \ref{item:mpc-sra-3} is true, there exists a $\delta^{\dagger} \in (0, x_{\min}]$ such that \eqref{eq:mpc-epsilon} holds for $Z=W_{x_{\min}}$ and $\varepsilon = \delta^{\dagger}$. Hence, 
\begin{equation*}
   W_{x_{\min}-\delta^{\dagger}} \laweq Y + (\delta_{\max} - x_{\min} + \delta^{\dagger})\id_{A_1} - (\delta_{\max} - x_{\min} + \delta^{\dagger})\id_{A_2} \laweq W_{x_{\min}} +  \delta^{\dagger}\id_{A_1}  - \delta^{\dagger}\id_{A_2} \succsim W_{x_{\min}} \succsim Y, \end{equation*}
which would indicate $(x_{\min}- \delta^{\dagger})\in\Xi$, thus contradicting the optimality of $x_{\min}$. Therefore, it must be that $x_{\min}=0$. This means that $\eqref{eq:mpc-epsilon}$ holds for $\varepsilon = \delta_{\max}$, and let us remark that $Y + \delta_{\max} \id_{A_1} - \delta_{\max} \id
_{A_2}$ takes the same value, $c$, on $A_1$ and $A_2$. This means that $\succsim$ is averse to symmetric splits, as defined in \cite{DS17}, and their Lemmas~1 and 2 establish that this aversion is equivalent to strong risk aversion.   \hfill$\square$

\subsection{The statement in  Proposition~\ref{lem:mpc-sra} fails for mean-preserving spreads}
\label{counter:mpc-sra}

Consider the two following statements:
\begin{enumerate}[label={\rm (\roman*)}, ref={\rm (\roman*)}]
\item\label{item:mps-sra-2} for every $X\in\mathcal{S}$, and equiprobable events $A_1, A_2\in\mathcal{F}$ such that $X=a_1$ on $A_1$ and $X=a_2$ on $A_2$, with $a_1 \leq a_2$, it holds that $X\succsim X - \delta \id_{A_1} + \delta \id_{A_2}$ for all  $\delta>0$;
\item\label{item:mps-sra-3} for every $X\in\mathcal{S}$, and equiprobable events $A_1, A_2\in\mathcal{F}$ such that $X=a_1$ on $A_1$ and $X=a_2$ on $A_2$, with $a_1 \leq a_2$, it holds that $X\succsim X - \delta^* \id_{A_1} + \delta^* \id_{A_2}$ for at least one $\delta^*>0$.
\end{enumerate}
The difference between the statement \ref{item:mps-sra-3} and the one in Proposition~\ref{lem:mpc-sra}\ref{item:mpc-sra-3} is that here we let the less dispersed random variable be fixed and search for a more dispersed random variable, while in Proposition~\ref{lem:mpc-sra}\ref{item:mpc-sra-3} we fix the more dispersed random variable and search for a less dispersed one. 
The following example shows that \ref{item:mps-sra-3} does not imply \ref{item:mps-sra-2} in general. 
\begin{example}
    Consider the continuous risk preference defined by letting
    \begin{equation*}
        X\succsim Y \quad \iff \quad \mathbb{E}[X^2] - \mathbb{E}[X^4] \geq \mathbb{E}[Y^2] - \mathbb{E}[Y^4], \quad X,Y\in\mathcal{S}.  
    \end{equation*}
    For any $X\in\mathcal{S}$, and any disjoint equiprobable non-null events $A_1, A_2\in\mathcal{F}$ such that $X=a_1$ on $A_1$ and $X=a_2$ on $A_2$, with $a_1 \leq a_2$, let $Y_{\delta} = X - \delta \id_{A_1} + \delta \id_{A_2}$ for all $\delta>0$. 
    As $\delta\to\infty$, it holds that $\mathbb{E}[Y^2_{\delta}] - \mathbb{E}[Y^4_{\delta}]\to -\infty$, and therefore there exists a $\delta^*$ such that $\mathbb{E}[X^2] - \mathbb{E}[X^4] \geq \mathbb{E}[Y^2_{\delta^*}] - \mathbb{E}[Y^4_{\delta^*}]$. Hence, $X\succsim Y_{\delta^*}$, so
    \ref{item:mps-sra-3} holds true. Meanwhile, consider the random variable $X$ degenerate at $0$. For equiprobable $A_1,A_2\in\mathcal{F}$ partitioning $\Omega$, let $Y = X - \delta\id_{A_1} + \delta\id_{A_2}$, with $\delta = 1/2$. We have
    \begin{equation}
        \mathbb{E}\left[X^2\right] - \mathbb{E}\left[X^4\right] = 0 \not \geq \frac{3}{16} = \mathbb{E}\left[Y^2\right] - \mathbb{E}\left[Y^4\right],
    \end{equation}
    showing that \ref{item:mpc-sra-1} does not hold. We remark that $\succsim$ is not strongly risk averse. 
\end{example}

\subsection{Proof of Theorem~\ref{th:full-strict}}
\label{proof:full-strict}

For sufficiency: Assume that for every interval $[b_1,b_2]\subseteq \mathbb{R}$, there is $C\in\mathcal{C}$ such that $[b_1,b_2]\subseteq J_C$. We argue that $\mathcal{C}$ satisfies second risk--insurance parity. If there is a contract $C\in\mathcal{C}$ for which $J_C = \mathbb{R}$, then this follows directly from Theorem~\ref{th:solo} and Proposition~\ref{prop:bigger-contract-sets}. Otherwise, the argument in the proof of Theorem~\ref{th:solo} can be reproduced almost identically: for each $Y\in\mathcal{S}$ considered, one can always choose $C\in\mathcal{C}$ to make $J_C$ large enough. First we remark this arbitrary choice of $J_C$
makes $\{h(x):x\in\mathbb{R}\}$ arbitrarily large. Secondly, the support of $Y$ is bounded, and one only requires $J_C$ to be large enough for $z_y$ to be well-defined in \eqref{eq:z-and-w-one-contract} for every $y\in\mathrm{supp}Y$.  

For necessity: Suppose that the condition on $\{J_C : C\in\mathcal{C}\}$ fails. This means that there exists $m\in\mathbb{R}$ such that $[-m,m]$ is not contained in any $J_C$, $C\in\mathcal{C}$. Hence, this means that, for any $C\in\mathcal{C}$, either $J_C\subseteq (-\infty,m)$ or $J_C\subseteq (-m,\infty)$. 
For any $Z\in \mathcal{S}$ and $C\in\mathcal{C}$, we have $\mathrm{supp} (Z+C(Z))\subseteq J_C$. Thus, the relation $\succsim_{K}$ defined in the proof of Theorem~\ref{th:solo} is an example of a continuous risk preference that is insurance propense to $\mathcal{C}$ but not strongly risk averse, and this means that $\mathcal{C}$ does not satisfy second risk--insurance parity.

\subsection{Proof of Proposition~\ref{prop:unbounded}}

For the implication strong risk aversion $\implies$ insurance propensity to $\mathcal{C}$: Because $Z$ and $C(Z)$ are countermonotonic for any $Z\in \mathcal{S}$ and $C\in\mathcal{C}$, it holds that $Z+C(Z) \leq_{\rm cx} Z + C(W)$ whenever $Z\laweq W$ by \citet[Corollary~3.28]{R13}, and the implication follows. 

For the implication insurance propensity to $\mathcal{C}$ $\implies$ strong risk aversion: Let $\succsim$ be any risk preference. Consider any $Y\in\mathcal{S}$ and equiprobable events $A_1, A_2\in\mathcal{F}$ such that $Y=a_1$ on $A_1$, and $Y=a_2$ on $A_2$, with $a_1 < a_2$. 
By Proposition~\ref{lem:mpc-sra}, it suffices to argue that $Y + \delta^* \id_{A_1} - \delta^*\id_{A_2} \succsim Y$ for at least one $\delta^* \in (0, (a_2-a_1)/2]$ to show that $\succsim$ is strongly risk averse. To this end, define $g:\mathbb{R}\to \mathbb{R}$ by letting $g(x) = x+I(x)$ for every $x\in\mathbb{R}$. Because $I$ is continuous and increasing, the function $g$ is continuous and strictly increasing, and thus the inverse $g^{-1}$ is well-defined, continuous and strictly increasing. 
Note that, for every $x\in\mathbb{R}$, 
\begin{equation}
    0 \leq I(g^{-1}(x+a_2-a_1)) - I(g^{-1}(x)) < a_2- a_1.
\label{eq:note-explicative}
\end{equation}
There must exist $x^*\in\mathbb{R}$ such that the difference in \eqref{eq:note-explicative} is nonzero, otherwise this would conflict with the requirement that $I$ is nonconstant. Let $s = g^{-1}(x^*)$ and write $t_s = g^{-1}(g(s) + a_2 -a_1)$. For any $y\in\mathbb{R}$, let $z_y\in\mathbb{R}$ be such that $y = \pi^* + z_y - I(z_y)$, where $\pi^* = a_1 - s + I(t_s)$ is the premium in the contract $C^*$ written $C^*(X) = \pi^* - I(X)$. (If $\pi^*\notin\Pi$, by taking a sequence of premiums in $\Pi$ converging to $\pi^*$ and keeping $Z,W$ fixed, the comparison follows by continuity of the risk preference.)  Note that $J_{C^*} = \mathbb{R}$ ensures $z_y$ to be well-defined for every $y\in\mathbb{R}$, and $J_{C^*} = \mathbb{R}$ holds by assumption that all contracts share the same indemnity function. 
Then, define random variables $Z$ and $W$ as in \eqref{eq:z-and-w-one-contract}. 
Taking $\delta^* = I(t_s) - I(s)$, we have
\begin{equation*}
    Y+ \delta^*\id_{A_1} - \delta^*\id_{A_2} \laweq Z + C^*(Z) \succsim Z + C^*(W) \laweq Y,
\end{equation*}
where equality in distribution relations are valid since identities \eqref{eq:proofth4-id1} and \eqref{eq:proofth4-id2} hold from the similar construction (see Table~\ref{tab:th4} for reference). 
On the legitimacy of the choice of $\delta^*$, see the proof of Theorem~\ref{th:solo}. Hence, by Proposition~\ref{lem:mpc-sra}, the risk preference $\succsim$ is strongly risk averse. \hfill$\square$

\subsection{Proof of Proposition~\ref{prop:charac-handles}}
\label{proof:charac-handles}

For \ref{item:charac-left-handle}:
We first consider necessity. Assume that $X\le_{\rm lhcx} Y$ and $X\not\laweq Y$. The case where $X$ is degenerate is direct; assume further that $X$ is not degenerate. The condition that $\E[X]=\E[Y]$ and $\min X > \min Y$ follow directly from the definition of left-handle mean-preserving spreads. Let $\{A_i : i\in [n]\}$ be a partition of $\Omega$ such that $X^{\rm co}$ and $Y^{\rm co}$ are both constant on each $A_i$, choosing $n\in\mathbb{N}$ accordingly. Let $x_i=X^{\rm co}(A_i)$ and $y_i=Y^{\rm co}(A_i)$ for all $i\in [n]$. Because $X^{\rm co}$ and $Y^{\rm co}$ are comonotonic, we can assume that $x_1\le x_2\le \dots \le x_n$ and $y_1\le y_2\le \dots \le y_n$. Choose $k^{*}$ such that $x_{k^*}=\min X$ and $x_{k^*+1}>\min X$. If the property of $\p(X\le t| X>\min X)\ge \p(Y^{\rm co}\le t|X^{\rm co}>\min X)$ for all $t\in\R$
does not hold, then $x_t>y_t$ for some $t>k^*$. This leads to a contradiction, as the left-handle mean-preserving spread does not strictly reduce the value for any point other than the minimal one. Next, we address sufficiency. Suppose that $\E[X]=\E[Y]$, $\min X> \min Y$, and $\p(X\le t| X>\min X)\ge \p(Y^{\rm co}\le t|X^{\rm co}>\min X)$ for all $t\in\R$. 
Apply the same notation in the proof of the necessity part: we have $x_1=x_2=\dots=x_{k^*}> y_1$ and $x_i\le y_i$ for all $i>k^*$. Let $r^*$ be the index such that $y_{r*}< x_1$ and $y_{r^*+1}\ge x_1$. Because $\mathbb{E}[X]=\mathbb{E}[Y]$, it holds that $\sum_{i=1}^{r^*} (x_i-y_i) \mathbb{P}(A_i) = \sum_{j=r^*+1}^{n} (y_j-x_j)\mathbb{P}(A_j)$; let these quantities be represented by $K$. For all $z\in(0,K]$,  define
\begin{align*}
    \zeta(z) &= \max\left\{k\in[r^*]: \sum_{i=1}^{r^*} (x_1-\max(y_i,y_k))\mathbb{P}(A_i) \geq z\right\} ;\\
    \xi(z) &= \min\left\{k\in\{r^*+1,\ldots, n\}: \sum_{j=r^*}^{k} (y_j-x_j) \mathbb{P}(A_j)> z\right\}.
\end{align*}
Let $z_1 < z_2 < \ldots < z_{\ell} = K$ denote the $\ell$  distinct values of the sums appearing in these definitions, and $z_0 = 0$. 
Consider the following sequence of random variables: $X_0=X$ and
\begin{equation}
\label{eq:mps-charac}
    X_i = X_{i-1} - \frac{z_i-z_{i-1}}{\sum_{j=1}^{\zeta(z_i)}\mathbb{P}(A_j)} \id_{\bigcup_{j=1}^{\zeta(z_i)}A_{j}} + \frac{z_i-z_{i-1}}{\mathbb{P}(A_{\xi(z_i)})} \id_{A_{\xi(z_i)}}, \quad \text{for every }i\in[\ell],
\end{equation}
Then, we have $X_{\ell} \laweq Y$, and \eqref{eq:mps-charac} is a left-handle mean-preserving spread for every $i\in [\ell]$. Hence, $X\le_{\rm lhcx} Y$. 
The proof for item \ref{item:charac-right-handle} is similar and thus omitted. 
\hfill$\square$

\subsection{Proof of Proposition~\ref{prop:handlecontinuity}}
\label{proof:handlecontinuity}

We prove the claim for $\le_{\rm lhcx}$; the case of $\le_{\rm rhcx}$ then follows directly by Corollary~\ref{co:r&lhcx}. 
Consider two sequences of random variables $\{X_n : n\in\mathbb{N}\}\subseteq \mathcal{S}$ and $\{Y_n : n\in\mathbb{N}\}\subseteq \mathcal{S}$ such that $X_n \stackrel{\rm B}{\to} X$ and $Y_n \stackrel{\rm B}{\to} Y$ for some $X,Y\in\mathcal{S}$, and $X_n\le_{\rm lhcx}Y_n$ for every $n\in\mathbb{N}$. We aim to prove that $X\le_{\rm lhcx} Y$. Applying Proposition~\ref{prop:charac-handles}\ref{item:charac-left-handle},
it is sufficient to verify that $\min X\ge \min Y$ and $X^{\rm co}\le Y^{\rm co}$ on $\{X^{\rm co}>\min X\}$, where $(X^{\rm co},Y^{\rm co})$ are comonotonic with $X^{\rm co}\laweq X$ and $Y^{\rm co}\laweq Y$.

First, it is well known that convex order is closed under bounded convergence. Note that left-handle convex order is stronger than convex order, and thus, we have $X\le_{\rm cx}Y$. This implies $\min X\ge \min Y$. Next, we prove $X^{\rm co}\le Y^{\rm co}$ on $\{X^{\rm co}>\min X\}$.
Let $V$ be a uniform random variable on $[0,1]$; see \citet[Proposition A.31]{FS16} for its existence on an atomless space. Denote by $q_Z$ the quantile function of $Z\in\mathcal S$, i.e., $q_Z(s)=\inf\{t: \p(Z\le t)\ge s\}$ for $s\in(0,1)$. It holds that $(q_X(V),q_Y(V))$ are comonotonic. Further, denote by $\alpha_0=\p(X=\min X)$.
It suffices to verify that for any $\alpha,\beta\in(0,1)$ with $\beta>\alpha>\alpha_0$, 
\begin{align}\label{eq-2025-1209-1}
\int_\alpha^\beta q_{X}(s)\d s\le \int_\alpha^\beta q_{Y}(s)\d s.
\end{align}
To see this, because 
$X_n \stackrel{\rm B}{\to} X$ and $Y_n \stackrel{\rm B}{\to} Y$, we have $q_{X_n}\to q_{X}$ and $q_{Y_n}\to q_{Y}$ bounded almost surely. Hence, it is straightforward to verify that for large enough $n$, $q_{X_n}(s)>\min X_n$ for all $s\in[\alpha,\beta]$. It follows from $X_n\le_{\rm lhcx} Y_n$ and Proposition~\ref{prop:charac-handles}\ref{item:charac-left-handle}  that $q_{X_n}(s)\le q_{Y_n}(s)$ for all $s\in[\alpha,\beta]$. Therefore, 
\begin{align*}
\int_\alpha^\beta q_{X}(s)\d s
=\lim_{n\to\infty} \int_\alpha^\beta q_{X_n}(s)\d s \le \lim_{n\to\infty}\int_\alpha^\beta q_{Y_n}(s)\d s=\int_\alpha^\beta q_{Y}(s)\d s.
\end{align*}
Hence, \eqref{eq-2025-1209-1} has been verified, and this completes the proof.\hfill$\square$

\subsection{Proof of Theorem~\ref{th:right-handle-r-i-p-if}}
\label{proof:right-handle-r-i-p}

Suppose that $\mathcal{C}$ does not contain full-indemnity contracts; their possible inclusion follows directly from Theorem~\ref{th:1st-r-i-p-standard} and the fact that weak risk aversion is weaker than right-handle risk aversion.  

\textbf{For standard contract sets:} Suppose that 
    \begin{align*}
\mathcal C=\{C_{\pi,d}: X\mapsto \pi - (X-d)_+\mid(\pi,d)\in\Theta\},
    \end{align*}
and $\{\pi + d : (\pi,d)\in\Theta \}$ is dense in $\mathbb{R}$.
We first argue that, for any continuous risk preference $\succsim$, right-handle risk aversion implies insurance propensity to $\mathcal{C}$.
Let $X,Y\in\mathcal S$ with $X\laweq Y$. For $C_{\pi,d}\in\mathcal C$, we aim to verify that $X+C_{\pi,d}(X)\succsim X+C_{\pi,d}(Y)$. The case that $X\le d$ is trivial. We assume that $\p(X>d)>0$. 
Define
\begin{align*}
&Z_1=X\wedge d-X+Y-d,\quad Z_2=X-X\wedge d \quad {\rm and}~Z_3=X-Y;\\
&A_1=\{Y>d,X<Y\},\quad A_2=\{Y\leq d,X>d\} \quad {\rm and}~A_3=\{X>Y>d\}.
\end{align*}
It is straightforward to check that $Z_1$, $Z_2$ and $Z_3$ are nonnegative on $A_1$, $A_2$ and $A_3$, respectively. 
By standard calculation, we have
\begin{align}
X+C_{\pi,d}(Y)=\pi+X-(Y-d)_+ &=\pi+X\wedge d-Z_1\id_{A_1}+Z_2\id_{A_2}+Z_3\id_{A_3}\notag\\
&=X+C_{\pi,d}(X)-Z_1\id_{A_1}+Z_2\id_{A_2}+Z_3\id_{A_3}.\label{eq:pitorho}
\end{align}
Since $X(\omega)+C_{\pi,d}(X(\omega))=\max\{X+C_{\pi,d}(X)\}$ for all $\omega\in A_2\cup A_3$, it follows from \eqref{eq:pitorho} and Proposition~\ref{prop:charac-handles} that $X+C_{\pi,d}(X)\le_{\rm rhcx} X+C_{\pi,d}(Y)$. Therefore, right-handle risk aversion
implies $X+C_{\pi,d}(X)\succsim X+C_{\pi,d}(Y)$.

Next, we prove that, for any risk preference $\succsim$, insurance propensity to $\mathcal{C}$ implies right-handle risk aversion. 
Consider any $X\in \mathcal{S}$, and suppose that $\max X = \pi^* + d^*$ for some $(\pi^*,d^*)\in\Theta$. (The case where expressing $\max X$ in such a way is impossible will follow directly due to the density of $\{\pi + d : (\pi,d)\in\Theta \}$, allowing to consider a sequence $X_n \to X$ for which this is verified, and evoking continuity of risk preferences.)
Let $Y\in \mathcal{S}$ result from a right-handle mean-preserving spread of $X$: for some $A_1,A_2\in\mathcal{F}$, $\delta_1,\delta_2>0$, we have $Y=X-\delta_1 \id_{A_1} + \delta_2 \id_{A_2}$ with $\delta_1 \p(A_1) = \delta_2 \p(A_2)$ and $X(\omega) = \max X$ for all $\omega\in A_2$.  We will prove that $X \succsim Y$. The general case for $X\le_{\rm rhcx} Y$ follows by transitivity. 

 For now, let us assume that $\mathbb{P}(A_2) = k\mathbb{P}(A_1)$ for some $k\in\mathbb{N}$. This implies $k\delta_2 = \delta_1$.
Partition $A_2$ into $k$ equiprobable events and denote these events $A_2^{(1)},\ldots, A_2^{(k)}$. Let $Z = X - \pi^* +  \delta_2 \sum_{i=1}^{k} i \id_{A_2^{(i)}}$.

Because $A_1$ and $A_2^{(1)}$ are equiprobable and the probability space is atomless, 
we can partition $A_1$ into finitely many events on which $X$ is constant and partition $A_2^{(1)}$ into events with the same respective probabilities.  Define $\xi:A_2^{(1)}\to A_1$ by choosing one representative point in each part of $A_1$, and sending the entire corresponding part of $A_2^{(1)}$ to that representative.
Then, define $W\in\mathcal{S}$ such that
\begin{equation}
    W(\omega) = \begin{cases}
        X(\omega ) - \pi^*, & \omega\in\Omega\backslash (A_1\cup A_2),\\
        \max X - \pi^* + k\delta_2, &\omega \in A_1,\\
        X(\xi(\omega)) - \pi^*, & \omega\in A_2^{(1)},\\
        \max X - \pi^* + (i-1)\delta_2,  &\omega\in A_2^{(i)}, \, i\in \{2,\ldots,k\}, \text{if }k>1.
    \end{cases}
    \label{eq:megaW}
\end{equation}
for all $i\in \{2,\ldots,k\}$ if $k>1$. Thus, $Z\laweq W$ and 
    $C_{\pi^*,d^*}(W) = C_{\pi^*,d^*}(Z) - \delta_1 \id_{A_1} + \delta_2 \id_{A_2}$; see Table~\ref{tab:deductible} for an illustration. Recall that $\pi^* + d^* = \max X$. 
We have
\begin{align}\label{eq-20251126-1}
Z + C_{\pi^*,d^*}(Z)&=Z+\pi^*-(Z-d^*)_+ =(\pi^*+Z)\wedge (\pi^*+d^*)\notag\\
&=\left(X+\delta_2\sum_{i=1}^ki\id_{A_2^{(i)}}\right)\wedge \max X=X,
\end{align}
where the last equality holds because $X(\omega)=\max X$ for all $\omega\in\bigcup_{i=1}^k A_2^{(i)}$. Therefore,
\begin{align*}
    X = Z + C_{\pi^*,d^*}(Z) \succsim Z + C_{\pi^*,d^*}(W) = Z + C_{\pi^{*}, d^*}(Z) - \delta_1 \id_{A_1} + \delta_2 \id_{A_2} =
    X- \delta_1 \id_{A_1} + \delta_2 \id_{A_2}=Y, 
\end{align*}
where the risk preference is established from the insurance propensity assumption. 

\begin{table}
\centering
\scriptsize\resizebox{\textwidth}{!}{
\begin{tabular}{cccccccc}
\toprule
 &  $\Omega\backslash(A_1\cup A_2)$ & $A_1$ & $A_2^{(1)}$ & $A_2^{(2)}$ & $\cdots$ & $A_2^{(k)}$\\
\midrule
$Z$ &  $X(\omega)- \pi^*$ & $X(\omega)- \pi^*$ & $\max X - \pi^* + \delta_2$ & $\max X- \pi^* + 2\delta_2$ & $\cdots$ & $\max X- \pi^* + k\delta_2$\\
$W$ & $X(\omega)- \pi^*$ &  $\max X - \pi^* + k\delta_2$ & $X(\xi(\omega))- \pi^*$& $\max X - \pi^* + \delta_2$& $\cdots$& $\max X - \pi^* + (k-1)\delta_2$\\
$C_{\pi^*,d^*}(Z)$ & $\pi^*$ & $\pi^*$ & $\pi^*-\delta_2$& $\pi^*-2\delta_2$ & $\cdots$& $\pi^*-k\delta_2$\\
$C_{\pi^*,d^*}(W)$ & $\pi^*$ & $\pi^*-k \delta_2$& $\pi^*$&  $\pi^*-\delta_2$&  $\cdots$& $\pi^*-(k-1)\delta_2$ \\
$Z + C_{\pi^*,d^*}(Z)$ & $X(\omega)$ & $X(\omega)$ & $\max X$ & $\max X$ & $\cdots$& $\max X$\\
$Z + C_{\pi^*,d^*}(W)$  & $X(\omega)$ & $X(\omega)-k\delta_2$ & $\max X + \delta_2$ & $\max X + \delta_2$ & $\cdots$& $\max X + \delta_2$\\
\bottomrule
\end{tabular}}
\caption{Illustration of the contracts' construction for the proof of Theorem~\ref{th:right-handle-r-i-p-if}.   
} 
 \label{tab:deductible}   
\end{table}

Next, consider the case that $\p(A_1)/\p(A_2)$ is a rational number that is equal to $m/n$, where $m,n\in\N$. Then, we can partition $A_1$ into $m$ equally probable events, denoted by $\{A_1^{(i)}:i\in[m]\}$. Let $X_0=X$, and for $i\in[m]$, define
\begin{align*}
X_i=X_{i-1}- {\delta_1}\id_{A_1^{(i)}}+\frac{\delta_2}{m}\id_{A_2}.
\end{align*}
It is straightforward to see that $X_m=Y$, and $X_i$ is a right-handle mean preserving spread of $X_{i-1}$ for $i\in[m]$. Noting that $\p(A_2)=n\p(A_1^{(i)})$ for $i\in[m]$, and it follows from the previous arguments that $X_{i-1}\succsim X_i$. The transitivity of $\succsim$ further implies $X\succsim Y$. 

Let us now consider the general case that $\p(A_1)/\p(A_2)\in\R$. There exist sequences of rational numbers $\{p_i : i\in\N\}$ and  $\{q_i:i\in\N\}$ such that $p_i\uparrow \p(A_1)$ and $q_i\uparrow \p(A_2)$, respectively. Let $\{A_1^{(i)}: i\in\N\}\subseteq\mathcal F$ and $\{A_2^{(i)} : i\in\N\}\subseteq\mathcal F$ be two sequences satisfying $A_1^{(i)}\uparrow A_1$, $A_2^{(i)}\uparrow A_2$, and $\p(A_1^{(i)})=p_i$ and $\p(A_2^{(i)})=q_i$ for all $i\in\N$. Define
\begin{equation*}
Y_i=X-\delta_1^{(i)}\id_{A_1^{(i)}}+\delta_2^{(i)}\id_{A_2^{(i)}},
\end{equation*}
where the sequences $\{\delta_1^{(i)} : i\in\N\}$ and  $\{\delta_2^{(i)} : i\in\N\}$ satisfy 
$\delta_1^{(i)}\to \delta_1$, $\delta_2^{(i)}\to \delta_2$, and $\delta_1^{(i)}\p(A_1^{(i)})=\delta_2^{(i)}\p(A_2^{(i)})$ for all $i\in\N$. By the previous arguments for the case that $\p(A_1)/\p(A_2)$ is a rational number, we have $X\succsim Y_i$ for all $i\in\N$. The fact that $Y_i\stackrel{\rm B}{\to} Y$ and the continuity of $\succsim$ yield $X\succsim Y$.

\textbf{For $\rho$-priced contract sets.}
    For the ``if" statement: Take any premium principle $\rho$. Suppose that 
    \begin{equation*}
\mathcal C=\{C_{d}: X\mapsto \rho((X-d)_+) - (X-d)_+\mid d\in\Pi\},
    \end{equation*}
and $\Pi$ is dense in $\mathbb{R}$. The argument supporting that right-handle risk aversion implies insurance propensity to $\mathcal{C}$ is the same as for standard contract sets, but switching $\pi$ for $\rho((X-d)_+)$ in \eqref{eq:pitorho}. 
For the converse implication, consider any $Y\in\mathcal{S}$ resulting from a right-handle mean-preserving spread of some $X \in \mathcal{S}$: 
for some $A_1,A_2\in\mathcal{F}$, $\delta_1,\delta_2>0$, we have $Y=X-\delta_1 \id_{A_1} + \delta_2 \id_{A_2}$ with $\delta_1 \p(A_1) = \delta_2 \p(A_2)$ and $X(\omega) = \max X$ for all $\omega\in A_2$. We will prove that $X \succsim Y$. The general case for $X\le_{\rm rhcx} Y$ follows by transitivity. 
Assume $\mathbb{P}(A_2) = k\mathbb{P}(A_1)$, for some $k\in\mathbb{N}$. Partition $A_2$ into $k$ equiprobable events and denote them by $A_2^{(1)},\ldots,A_2^{(k)}$. 

Let $\pi^* = \rho( (X + \delta_2 \sum_{i=1}^k i \id_{A_2^{(i)}} - \max X)_+)$. Consider the sequence $\{d_i: i\in\mathbb{N}\}\subseteq \Pi$ converging to $d^* = \max X - \pi^*$; such a limit exists because $\Pi$ is dense in $\mathbb{R}$. 
Write $Z = X - \pi^* + \delta_2 \sum_{i=1}^k i \id_{A_2^{(i)}}$, and define $W$ as in \eqref{eq:megaW}. Therefore, $Z\laweq W$ and $C_{d^*}(W) = C_{d^*}(Z) - \delta_1 \id_{A_1} + \delta_2 \id_{A_2}$; see above.

Let $Z_i = Z+ d_i - d^*$ and $W_i = W + d_i - d^*$ for every $i\in\mathbb{N}$. 
Note that, for every $i\in\mathbb{N}$,
\begin{align}
    \rho((Z_i - d_i)_+) = \rho((Z-d^*)_+) = \rho\left(\left(X + \delta_2 \sum_{i=1}^k i\id_{A_2^{(i)}} - \max X \right)_+\right) = \pi^*.\label{eq:proof-rho-ded}
    \end{align}
By the insurance propensity assumption, and noting that $\rho((Z_i - d_i)_+) = \rho((W_i - d_i)_+)$ by law invariance, we establish $Z_i + C_{d_i}(Z_i) \succsim Z_i + C_{d_i}(W_i)$ for every $i\in\mathbb{N}$; by continuity of risk preferences, we have $Z + C_{d^*}(Z) \succsim Z + C_{d^*}(W)$ because $Z_i + C_{d_i}(Z_i) \stackrel{\rm B}{\to} Z + C_{d^*}(Z)$ and $Z_i + C_{d_i}(W_i) \stackrel{\rm B}{\to} Z + C_{d^*}(W)$
by \eqref{eq:proof-rho-ded}. 

Sharing the similar arguments of \eqref{eq-20251126-1} yields 
$
    Z + C_{d^*}(Z)=X.
$
It follows that
\begin{align*}
    X = Z + C_{d^*}(Z) \succsim Z + C_{d^*}(W) = Z + C_{d^*}(Z) - \delta_1\id_{A_1} + \delta_2\id_{A_2} =
    X- \delta_1\id_{A_1} + \delta_2\id_{A_2}=Y.
\end{align*}
The remaining arguments allowing to extend the results to general $A_1$ and $A_2$ are the same as for standard contract sets. \hfill$\square$

\subsection{Proof of Proposition~\ref{th:right-handle-r-i-p-shape}}

We argue that the implication right-handle risk aversion$\implies$insurance propensity to $\mathcal{C}$ only holds true if $\mathcal{C}$ contains deductible-only or full-indemnity contracts exclusively. Define the risk preference $\succsim$ as such
\begin{equation*}
    X \succsim Y \iff X \le_{\rm rhcx} Y.
\end{equation*}
By Proposition~\ref{prop:handlecontinuity}, $\succsim$ is continuous. Therefore, for $\mathcal{C}$ to satisfy right-handle risk--insurance parity,  we must have
\begin{equation}
Z + C(Z) \le_{\rm rhcx} Z + C(W), \quad \text{for every }C\in\mathcal{C}, \text{ and }Z,W\in\mathcal{S} \text{ with }Z\laweq W.
\label{eq:czw-right-handle}
\end{equation}
 Otherwise, $\succsim$ is right-handle risk averse but not insurance propense to $\mathcal{C}$. 
Below we aim to verify that  if there is $C\in\mathcal C$ that is neither a deductible-only nor a full-indemnity contract, then we can construct $Z,W\in\mathcal S$ with $Z\laweq W$ such that $Z + C(Z) \not\le_{\rm rhcx} Z + C(W)$. It is straightforward to check that a contract $C$ is either a deductible-only or a full-indemnity contract if and only if, for any $y\in\R$, its indemnity function satisfies that $(I(y)-I(x))/(y-x)>0$ for some $x<y$ implies $(I(z)-I(y))/(z-y)=1$ for all $z>y$.\footnote{This equivalence holds only if we also allow for an additive constant in the indemnity functions of deductible or full-indemnity contracts; i.e., the indemnity functions of the contracts are either in the form $I(x) = (x-d)_+ +a$ or $I(x) = 
x+a$. This constant, however, is automatically absorbed into the premium for standard contracts (see Section~\ref{sect:contracts}) and, in the main text, we explicitly account for it for $\rho$-priced contracts.} Then, suppose that there is $C\in\mathcal C$ such that its indemnity function satisfies
\begin{align}\label{eq-20251201-1}
\frac{I(y_0)-I(x_0)}{y_0-x_0}\in(0,1]\quad{\rm and}\quad \frac{I(z_0)-I(y_0)}{z_0-y_0}\in[0,1)\quad\text{for some }x_0<y_0<z_0,
\end{align}
and this is consistent with the LC assumption; we show that this contradicts right-handle risk--aversion. Define $Z,W\in\mathcal{S}$ with joint probability masses 
$\p(Z=x_0,W=z_0)=\p(Z=y_0,W=x_0)=\p(Z=z_0,W=y_0)=1/3$.
Obviously, we have $Z\laweq W$.
Write $X=Z+C(Z)$ and $Y=Z+C(W)$; it holds that 
$
\p(X=x_i)=\p(Y=y_i)=\frac{1}{3}$ for $i=1,2,3$,
where 
\begin{align*}
&x_1=x_0 + \pi -I(x_0),\quad x_2=y_0 + \pi -I(y_0),\quad x_3=z_0 + \pi -I(z_0);\\
&y_1=x_0 + \pi -I(z_0),\quad y_2=y_0 + \pi -I(x_0),\quad y_3=z_0 + \pi -I(y_0).
\end{align*}
It follows from \eqref{eq-20251201-1} that
\begin{align*}
x_1\le x_2< x_3,~ y_2>x_2~{\rm and}~y_3>x_2.
\end{align*}
This implies $X^{\rm co}\not\ge Y^{\rm co}$ on the event $\{X^{\rm co}<\max X\}$, where $(X^{\rm co}, Y^{\rm co})$ are comonotonic with $X^{\rm co}\laweq X$ and $Y^{\rm co}\laweq Y$. Proposition \ref{prop:charac-handles}\ref{item:charac-right-handle} yields $X\not\le_{\rm rhcx} Y$, leading to a contradiction. Therefore, we have concluded that $\mathcal{C}$ only contains deductible-only or full-indemnity contracts. For a $\rho$-priced contract, replace $\pi$ by $\rho(I(Z))$ above and the argument remains the same.  
\hfill$\square$

\subsection{Proof of Proposition~\ref{prop:right-handle-r-i-p-double-equiv}}

From Theorem~\ref{th:right-handle-r-i-p-if} and Proposition~\ref{th:right-handle-r-i-p-shape}, 
it only remains to argue that the set of price parameters, denoted by $\Pi$, must be dense in $\mathbb{R}$. Suppose $\Pi$ is not dense in $\mathbb{R}$, and we will show that there exists a continuous risk preference for which insurance propensity to $\mathcal{C}$ does not imply right-handle risk aversion. Define the risk preference $\succsim$ by 
\begin{equation*}
    X\succsim Y \iff (X\le_{\rm rhcx} Y \text{ and }\max X \text{ is in the closure of } \Pi),\text{ or either }X=\mathbb{E}[Y] \text{ or } X\laweq Y.
\end{equation*}
It is straightforward to verify that $\succsim$ is law invariant and transitive. In the proof of Proposition~\ref{prop:handlecontinuity}, we argued that for sequences $\{X_n : n\in\mathbb{N}\}\subseteq \mathcal{S}$ and $\{Y_n : n\in\mathbb{N}\}\subseteq\mathcal{S}$ such that $X_n\stackrel{\rm B}{\to}X$, $Y_n\stackrel{\rm B}{\to} Y$ and $X_n \le_{\rm rhcx} Y_n$ for every $n\in\mathbb{N}$, either $\max X_n \to \max X$ or $X\laweq Y$. Thus, $\succsim$ is continuous. Consider any $Z\in\mathcal{S}$ and $C_\theta\in\mathcal{C}$ that has the form $C_\theta(X)=\pi-(X-d)_+$ with $\pi+d=\theta$. If $\max Z < d$, then $Z + C_{\theta}(Z) = Z + C_{\theta}(W) = Z+ \pi$ for every $W\in\mathcal{S}$, $W\laweq Z$; and $Z+C_{\theta}(Z)\succsim Z+ C_{\theta}(W)$. If $\max Z \geq d$, we have $\theta = \max (Z+C_{\theta}(Z))$. We also previously argued, in the proof of the ``if" direction, that $Z+C_{\theta}(Z) \le_{\rm rhcx} Z+C_{\theta}(W)$ for every $W\in\mathcal{S}$ with $W\laweq Z$. Thus, $Z+C(Z)\succsim Z+C(W)$ also for this case. The continuous risk preference $\succsim$ is therefore insurance propense to $\mathcal{C}$. However, by definition, it is not right-handle risk averse.  
\hfill$\square$

\subsection{A counterexample to Proposition~\ref{prop:right-handle-r-i-p-double-equiv} in the absence of LC}
\label{sec:ex-non-LC}

\begin{example}
\label{ex:non-LC}
Consider the standard contract set $\mathcal{C}$ comprising all deductible-indemnity contracts, for every combination of premium $\pi\in\mathbb{R}$ and deductible $d\in\mathbb{R}$. By Theorem~\ref{th:right-handle-r-i-p-if}, $\mathcal{C}$ satisfies right-handle risk--insurance parity. 
Define $I_* : x \mapsto x + (1-x)\id_{\{0<x<1\}}$; note that $I_*$ is not LC. 
Define $C_* : X \mapsto \pi - I_*(X)$, where $\pi\in\mathbb{R}$, and let $\mathcal{C}^{\dagger} = \mathcal{C}\cup\{C_*\}$. We will show that $\mathcal{C}^{\dagger}$ also satisfies right-handle risk--insurance parity. Because $\mathcal{C}\subseteq \mathcal{C}^{\dagger}$, the claim that, for every continuous risk preference, insurance propense to $\mathcal{C}^{\dagger}$ implies right-handle risk aversion holds directly. We must show that right-handle risk aversion implies insurance propensity to $\{C_*\}$. We aim to verify that 
$
    X + C_*(X) \succsim X + C_*(Y)$ for every $X,Y\in\mathcal{S}$ with $X\laweq Y$.  
The case where $X,Y$ are degenerate is direct from law invariance; suppose next that they are not degenerate. 

The case where $\mathbb{P}(0<X<1)=0$  follows directly from the ``if" direction of Theorem~\ref{th:1st-r-i-p-standard} and the fact that right-handle risk aversion is stronger than weak risk aversion. 
If $\mathbb{P}(0<X<1) =1 $, then $X + C_*(X) \laweq X + C_*(Y)$ and the claim holds by law invariance of risk preferences.  

If $0< \mathbb{P}(0<X<1) < 1$, then  $\max(X+C_*(X)) = \pi$. Moreover, for any $\pi-1 \leq z <\pi$, we have
\begin{align*}
    \mathbb{P}(X+C_*(Y)\leq z) \geq \mathbb{P}(X-1 \leq z - \pi, Y> 0) &\geq \mathbb{P}(X \leq z- \pi +1) - \mathbb{P}(Y\leq 0)\\
    &= \mathbb{P}(0< X \leq z - \pi+1)
    =\mathbb{P}(X + C_*(X) \leq z),
\end{align*}
and $\mathbb{P}(X + C_*(X) \leq z) = 0$ for $z<\pi-1$.
Since the $X+C_*(X)$ and $X+C_*(Y)$ have the same expectation, these CDF inequalities imply that either their distributions coincide or $\max(X+C_*(Y))>\pi$.
Therefore, by Proposition~\ref{prop:charac-handles}\ref{item:charac-right-handle}, we have 
 $X + C_*(X) \leq_{\rm rhcx} X + C_*(Y)$.
Right-handle risk aversion thus implies $X + C_*(X) \succsim X + C_*(Y)$. We conclude that the contract set $\mathcal{C}^{\dagger}$ satisfies right-handle risk--insurance parity despite the added non-deductible contract (which is also non-LC). 
\end{example}

\subsection{Proof of Theorem~\ref{th:left-handle-r-i-p-if}}
\label{proof:left-handle-r-i-p}
For any $X,Y\in \mathcal{S}$,  $X\le_{\rm lhcx}Y$ if and only if $-X \le_{\rm rhcx} -Y$; see Corollary~\ref{co:r&lhcx}. 
In the case of standard contract sets, note that, for any $Z\in\mathcal{S}$, 
$\pi,\lambda\in\mathbb{R}$, $$\pi - Z\wedge\lambda = - ( (\lambda - \pi) - ((-Z)-(-\lambda))_+),$$ that is: for $C$ a limit-only contract, $X\mapsto-C(-X)$ is a deductible-only contract. 
 For $\rho$-priced contract sets, we note that, for any $Z\in\mathcal{S}$, $\zeta\in\mathbb{R}$, it holds that 
\begin{align*}
   \rho(Z\wedge \zeta - \zeta) - (Z \wedge \zeta - \zeta)= - \left( -\rho(-((-Z)-(-\zeta))_+) - ((-Z) - (-\zeta))_+ \right).
\end{align*}
In other words, for a $\rho$-priced contract $C$ of the form specified in the theorem statement, $X\mapsto-C(-X)$ is a $\rho^*$-priced deductible-indemnity contract with deductible $-\zeta$, where $\rho^*(X) = -\rho(-X)$ for all $X\in\mathcal S$. 
In both cases, when re-expressing $-X$ in terms of an insurance contract, we have $-X = -Z + (-C(-(-Z)))$, and thus everything follows from the proof of Theorem~\ref{th:right-handle-r-i-p-if} by symmetry.

\subsection{Proof of Proposition~\ref{th:left-handle-r-i-p-shape}}

The result follows from Proposition~\ref{th:right-handle-r-i-p-shape} once we evoke the same symmetry argument as in the proof of Theorem~\ref{th:left-handle-r-i-p-if}.

\subsection{Proof of Proposition~\ref{th:left-handle-r-i-p-double-equiv}}

The result follows directly from Proposition~\ref{prop:right-handle-r-i-p-double-equiv} by symmetry. Note that for a limit-only contract written $C_{\pi, \lambda}(Z)=\pi-Z\wedge \lambda$ the contract $Z\mapsto-C_{\pi, \lambda}(-Z)$ is a deductible-only contract that has premium $\lambda - \pi$ and deductible $ -\lambda$; therefore, the condition for exhaustiveness of price parameters in Proposition~\ref{prop:right-handle-r-i-p-double-equiv}
 becomes simply one on premiums here, as the $\lambda$'s cancel out.

\subsection{Proof of Theorem~\ref{prop:yaari-right-handle}}
\label{sect:proof-yaari-right-handle}


For \ref{item:star-1}: We first show that a weighting function $h$ satisfies star-shapedness at $1$ if and only if the following condition holds: 
\begin{align}\label{eq-iffDTrhra}
\frac{h(y)-h(x)}{y-x} \leq \frac{1 - h(z)}{1-z},\quad \text{for all }  0\le x< y\le z< 1.
\end{align}
Suppose that \eqref{eq-iffDTrhra} holds. Then, letting $y=z$, we have
\begin{align}\label{eq-stareq}
h(z) \leq \frac{z-x}{1-x} + \frac{1-z}{1-x} h(x), \quad \text{for all }  0\le x<z< 1.
\end{align}
 Taking $\lambda = (z-x)/(1-x)$, \eqref{eq-stareq} becomes 
\begin{equation*}
    h(\lambda + (1-\lambda)x) \leq \lambda + (1-\lambda) h(x), \quad \text{for all }  x\in[0,1),\lambda\in(0,1), 
\end{equation*}
 which is exactly star-shapedness at $1$. Conversely, star-shapedness at $1$ implies \eqref{eq-stareq}. Thus, 
\begin{align*}
 \frac{h(y)-h(x)}{y-x} \leq \frac{1-h(y)}{1-y} \leq \frac{1-h(z)}{1-z}.
\end{align*}
where the inequalities follow directly from two distinct applications of \eqref{eq-stareq}.
We thus have shown that star-shapedness at $1$ is equivalent to the condition in \eqref{eq-iffDTrhra}.

Next, we argue that the agent is right-handle risk averse if and only if their weighting function satisfies \eqref{eq-iffDTrhra}. First, suppose that the agent is right-handle risk averse. Let $B_0,\ldots, B_3$ be a measurable partition of $\Omega$, and $p_i=\mathbb{P}(B_i)$. Define $X$ and $Y$ as such:
\begin{align*}
X&=0\id_{B_0}-\id_{B_1}-2\id_{B_2}-3\id_{B_3};\\ Y&=p_2\id_{B_0}-\id_{B_1}-(2+p_0)\id_{B_2}-3\id_{B_3}.
\end{align*}
It is straightforward to check that $X\le_{\rm rhcx} Y$, and  
\begin{align*}
&h(p_1+p_2+p_3) + h(p_2+p_3) + h(p_3) =U_h(-X) \\&\quad \ge U_h(-Y)
= -p_2 + (1+p_2)h(p_1+p_2+p_3) + (1+p_0)h(p_2+p_3) + (1-p_0)h(p_3).
\end{align*}
This implies
\begin{align*}
 \frac{h(p_2+p_3) - h(p_3)}{p_2}\le \frac{1 - h(p_1+p_2+p_3)}{1-(p_1+p_2+p_3)},\;\; \text{for all } p_1,p_2,p_3\in[0,1],p_2>0~{\rm with}~p_1+p_2+p_3 < 1.
\end{align*}
Let $x=p_3$, $y=p_2+p_3$ and $z=p_1+p_2+p_3$. The above equation reduces to \eqref{eq-iffDTrhra}.

Conversely, suppose that \eqref{eq-iffDTrhra} holds. The preference induced by the loss function $U_h$ in \eqref{eq:yaari} is translation invariant. For convenience, let us suppose $X,Y$ are nonpositive (so that $-X,-Y$ are nonnegative).  Let $Y= X - \delta_1 \id_{A_1} + \delta_2\id_{A_2}$ such that $\delta_1 \p(A_1) = \delta_2\p(A_2)$ and $X(\omega) = \max X$ for all $\omega\in A_2$; the general case where $X \leq_{\rm rhcx} Y$ will follow from law invariance, transitivity and continuity of the induced risk preference. 

Also, assume $X$ is constant on $A_1$; otherwise, partition $A_1$ into $A_1^{(1)}, \ldots, A_1^{(n)}$, $n\in\mathbb{N}$, where $X$ is constant on each and consider successive right-handle mean-preserving spreads with $\delta_1^{(i)} = \delta_1$, $A_2^{(i)} = A_2$ and $\delta_2^{(i)} = \delta_2 \mathbb{P}(A_1^{(i)})/\mathbb{P}(A_1)$ for every $i\in[n]$.  

Then, for any weighting function $h$ that is star-shaped at $1$, we have
\begin{align*}
    U_h(-X) - U_h(-Y) &= \int_{0}^{\max( -Y)} \big(h(\mathbb{P}(-X>x)) - h(\mathbb{P}(-Y>x)) \big)\mathrm{d}x  \\
    &= \int_{\max X}^{\max Y} \big(1 - h(\mathbb{P}(Y<x))\big) \mathrm{d}x  -  \int_{\min Y}^{\max X} \big(h(\mathbb{P}(Y<x)) - h(\mathbb{P}(X<x))\big) \mathrm{d}x\\
    &\geq \delta_2\big(1- h(1-\mathbb{P}(A_2))\big) - \delta_1 \sup_{x\in [\min Y,\max X]}\left(h(\mathbb{P}(Y<x)) - h(\mathbb{P}(X<x)) \right)\\
    &\geq \delta_2\big(1- h(1-\mathbb{P}(A_2))\big) - \delta_1 \frac{\mathbb{P}(A_1)}{\mathbb{P}(A_2)}  \big(1- h(1-\mathbb{P}(A_2))\big) = 0,
\end{align*}
where the second inequality follows from \eqref{eq-iffDTrhra} and the observation that $\mathbb{P}(X < x) \le \mathbb{P}(Y < x) \le 1- \mathbb{P}(A_2)$ for every $x\in[\min Y,\max X]$.

 For \ref{item:star-0}:  The proof follows a similar argument as for part \ref{item:star-1} and is thus omitted for brevity.
 The statement of item \ref{item:star-0-1} follows directly from \ref{item:star-1} and \ref{item:star-0}. \hfill$\square$


\subsection{Proof of Proposition~\ref{prop:L0-Linfty}}
\label{proof:L0-Linfty}
    The implication \ref{item:L1}$\implies$\ref{item:S} is direct. To prove the converse implication \ref{item:S}$\implies$\ref{item:L1},     
    take any $X,Y\in L^p$ such that $X\laweq Y$. If $I$ is LC, define 
    \begin{align*}
        X_n=2^{-n}\left\lfloor 2^n\max(\min(X, n), -n)\right\rfloor;\quad\quad Y_n=2^{-n}\left\lfloor 2^n\max(\min(Y,n), -n)\right\rfloor, \quad n\in\mathbb{N}.
    \end{align*}
    
If $I$ is bounded but not LC, 
 define
\begin{equation*}
A_{n,k,\ell} = \left\{x\in\mathbb R:\left\lfloor n\bigl((-n)\vee(x\wedge n)\bigr)\right\rfloor=k,\quad\lfloor nI(x)\rfloor=\ell
\right\},\qquad \text{for every }k,\ell\in\mathbb Z, \quad n\in\mathbb{N}.
\end{equation*}
For each nonempty $A_{n,k,\ell}$, choose $a_{n,k,\ell}\in A_{n,k,\ell}$ satisfying $|a_{n,k,\ell}|\le\inf_{x\in A_{n,k,\ell}}|x|+1/n.$
Define
\begin{equation*}
X_n=\sum_{\substack{k,\ell\in\mathbb{Z} \\ A_{n,k,\ell}\ne\varnothing}}
a_{n,k,\ell}\,\id_{\{X\in A_{n,k,\ell}\}};
\qquad
Y_n=\sum_{\substack{k,\ell\in\mathbb{Z} \\ A_{n,k,\ell}\ne\varnothing}}
a_{n,k,\ell}\,\id_{\{Y\in A_{n,k,\ell}\}}\end{equation*}
Both sums are finite, and the same $a_{n,k,\ell}$ are used for $X_n$ and $Y_n$. Consequently, $X\laweq Y$ implies $X_n\laweq Y_n$ for every $n\in\mathbb{N}$. Moreover, note that 
\begin{equation}
|I(X_n)-I(X)|\le\frac{1}{n},\qquad|X_n-X|\le\frac{1}{n} \text{ on }\{|X|<n\},
\quad \text{and }
|X_n|\le |X|+\frac{1}{n}. \label{eq:convergence-witness}
\end{equation}
        
   In either case,  
    we have $X_n,Y_n\in \mathcal{S}$, $X_n\laweq Y_n$, and so, 
    \begin{equation*}
        X_n + C(X_n) \succsim X_n + C(Y_n) \quad \quad  \text{for every }n\in\mathbb{N}.
    \end{equation*}
   If $p<\infty$, we have $X_n + C(X_n) \stackrel{L^p}{\to} X + C(X)$ and $X_n + C(Y_n) \stackrel{L^p}{\to} X + C(Y)$, where $\stackrel{L^p}{\to}$ means convergence with respect to the $L^p$ norm; in the case of $\rho$-priced contracts, $L^p$-continuity of $\rho$ assures that these convergence relations hold. For the case were $I$ is bounded, but not LC, the convergence relations are justified by \eqref{eq:convergence-witness}.  If $p=\infty$,  we have $X_n + C(X_n) \stackrel{\rm B}{\to} X + C(X)$ and $X_n + C(Y_n) \stackrel{\rm B}{\to} X + C(Y)$. Hence, $X+C(X) \succsim X+C(Y)$ by continuity of $\succsim$. 
\hfill$\square$

 \subsection{Proof of Proposition~\ref{prop:extension-to-Linfty}}
\label{proof:extension-to-Linfty}

Let $\succsim$ be continuous on $L^p$. Its restriction to
$\mathcal S$ is continuous under bounded convergence, and
Proposition~\ref{prop:L0-Linfty} equates insurance propensity
on the two spaces.

Take any $X,Y$ in $L^p$. For every $n\in\mathbb{N}$, let
\begin{equation*}
\Psi_n =
\left\{
\left(\frac{0}{2^n},\frac{1}{2^n}\right],
\left(\frac{1}{2^n},\frac{2}{2^n}\right],
\ldots,
\left(\frac{2^n-1}{2^n},\frac{2^n}{2^n}\right)
\right\}
\end{equation*}
and $\Pi_n=U^{-1}(\Psi_n)$ where $U$ is a uniformly
distributed random variable on $(\Omega,\mathcal F,\mathbb P)$.
Define $\mathcal F_n=\sigma(\Pi_n)$ and
\begin{equation*}
X_n=\mathbb E\left[F_X^{-1}(U)\mid\mathcal F_n\right];
\qquad
Y_n=\mathbb E\left[F_Y^{-1}(U)\mid\mathcal F_n\right].
\end{equation*}
\citet[proof of Theorem~3]{MMWW25}
showed that if $X\le_{\rm cx}Y$ on
$(\Omega,\mathcal F,\mathbb P)$, then $X_n\le_{\rm cx}Y_n$
on $(\Omega,\mathcal F_n,\mathbb P_{|\mathcal F_n})$
for every $n\in\mathbb N$.
If $p<\infty$, as
$\sigma(\bigcup_{n\in\mathbb N}\mathcal F_n)
=\sigma(U)$, we have
$X_n\stackrel{L^p}{\to}F_X^{-1}(U)\laweq X$ and
$Y_n\stackrel{L^p}{\to}F_Y^{-1}(U)\laweq Y$
by martingale convergence in $L^p$.
If $p=\infty$, as
$\sigma(\bigcup_{n\in\mathbb N}\mathcal F_n)
=\sigma(U)$, we have
$X_n\stackrel{\rm B}{\to}F_X^{-1}(U)\laweq X$ and
$Y_n\stackrel{\rm B}{\to}F_Y^{-1}(U)\laweq Y$
by martingale convergence.

Thus, strong risk aversion extends from $\mathcal S$ to
$L^p$ by continuity and law invariance; weak risk aversion
extends similarly, since expectations converge.
By law invariance, write an elementary right-handle spread
as $Y=X-\delta_1\id_{A_1}+\delta_2\id_{A_2}$ with
$X=\esssup X$ on $A_2$. Take conditional expectations
with respect to increasing finite sigma-fields containing
$A_1,A_2$ and generating $\sigma(X,A_1,A_2)$. These preserve
the spread identity and the extremal value on $A_2$.
Convergence, continuity and transitivity then extend
right-handle risk aversion; the left-handle case is symmetric.

Applying the results on $\mathcal S$ proves the sufficient
conditions, the corollaries and the contraction
characterization. The larger-contract-set result follows
directly from inclusion. For necessity, the witness
preferences from the corresponding proofs extend by the
same formulas to $L^p$. They remain continuous:
expectations are continuous, equality in distribution
and convex order are closed under $L^1$ convergence,
and their almost-sure bounding conditions are closed
under convergence in probability. Insurance propensity
extends by Proposition~\ref{prop:L0-Linfty}, and the same
 counterexamples remain valid.
\hfill$\square$

\section{Monotone risk aversion and monotone risk--insurance parity}
\label{sect:monotone}

In this section, we discuss monotone risk aversion and its relation with risk--insurance parity. 

Monotone risk aversion was introduced by \cite{Q92}; see \cite{C95} for a review on the concept. Similar to left- and right-handle risk aversion, this form of risk aversion also falls strictly in-between weak and strong risk aversion.  

\begin{definition} 
A risk preference $\succsim$ exhibits \emph{monotone risk aversion} if for every $X,Y,\Theta \in \mathcal{S}$ such that $Y = X + \Theta$, where $X,\Theta$ are comonotonic and $\mathbb{E}[\Theta]=0$, it holds that $X\succsim Y$. 
    A set of contracts $\mathcal{C}$ satisfies \emph{monotone risk--insurance parity} when, for all continuous risk preferences, insurance propensity to $\mathcal{C}$ is equivalent to monotone risk aversion.  
\end{definition}


We next show that no LC contract set satisfies monotone risk--insurance parity. 

\begin{proposition}
\label{th:monotone-r-i-p}
    No LC contract set satisfies monotone risk--insurance parity. 
\end{proposition}


While right-handle, left-handle and monotone risk aversion all belong to the middle grounds between weak and strong risk aversion, they are strictly different concepts of risk aversion. This is evident when juxtaposing Propositions~\ref{prop:right-handle-r-i-p-double-equiv}, \ref{th:left-handle-r-i-p-double-equiv} and~\ref{th:monotone-r-i-p}.

\proof{Proof of Proposition~\ref{th:monotone-r-i-p}.}
   It is evident that the set of full-indemnity contracts (write $\mathcal{C}^{\rm full}$) does not satisfy monotone risk--insurance parity because monotone risk aversion is stronger in general than weak risk aversion. Next, suppose that for some contract set $\mathcal{C} \not\subseteq \mathcal{C}^{\rm full}$, monotone risk aversion implies insurance propensity to $\mathcal{C}$. We will prove by contradiction that this is absurd.

Define the relation $\succsim$ by 
\begin{equation}
    X\succsim Y \iff \substack{ \text{There exists }\Theta\in\mathcal{S} \text{ such that}\\ Y\laweq X + \Theta,~\mathbb{E}[\Theta] = 0, \text{ and }X\text{ and }\Theta \text{ are comonotonic}.}
    \label{eq:proof-monotone-succsim}
\end{equation}
For any random vector $(X,Y)$, let $(X^{\rm co}, Y^{\rm co})$ be its comonotonic version, meaning $X\laweq X^{\rm co}$, $Y\laweq Y^{\rm co}$, and $(X^{\rm co}, Y^{\rm co})$ is comonotonic. 
Note that the right-hand side of \eqref{eq:proof-monotone-succsim} is equivalent to requiring that for $\Theta = Y^{\rm co} - X^{\rm co}$ with $\mathbb{E}[\Theta] = 0$, we have that $\Theta$ and $X^{\rm co}$ are comonotonic.   
It is straightforward to verify that $\succsim$ is transitive, law invariant, and continuous. It is also monotone risk averse by definition. Suppose that $\succsim$ is insurance propense to $\mathcal{C}$; we will show that this creates a contradiction.

Consider any $C\in\mathcal{C}\backslash \mathcal{C}^{\rm full}$, with premium $\pi$ and indemnity function $I$. Let $x,y \in\mathbb{R}$ satisfy $x < y $ and $x - I(x) \leq y - I(y)$, which is possible by increasingness, LC and the assumption that $I$ is not a full indemnity. By LC, there exists $\varepsilon>0$ such that 
\begin{equation}
I(y+\varepsilon) - I(x) < y - x.\label{eq:proof-monotone-4}
\end{equation}
Let $z = y + \varepsilon$, and note that $y - I(y) \leq z - I(z)$. 

Define $Z^*,W^*\in\mathcal{S}$ with joint probability masses 
    $\mathbb{P}(Z^* = x, W^*= x) = \mathbb{P}(Z^* = y, W^*= z) = \mathbb{P}(Z^* = z, W^*= y) = 1/3$.
Write $X^* = Z^* + C(Z^*)$, $Y^* = Z^* + C(W^*)$ and $\Theta^* = Y^{*{\rm co}} - X^{*{\rm co}}$. By insurance propensity, $X^* \succsim Y^*$, so $\Theta^*$ and $X^{*{\rm co}}$ must be comonotonic. 
Since $ x - I(x) < y - I(z) $ by \eqref{eq:proof-monotone-4}, we have
\begin{align*}
    &\mathbb{P}(X^{* {\rm co}} = x + \pi - I(x), Y^{* {\rm co}}= x + \pi - I(x)) \\&= \mathbb{P}(X^{* {\rm co}} = y + \pi - I(y), Y^{* {\rm co}}= y + \pi - I(z))\\ &= \mathbb{P}(X^{* {\rm co}} = z + \pi - I(z), Y^{* {\rm co}}= z + \pi - I(y)) = \frac{1}{3}.
\end{align*}
As a result,
\begin{align*}
    &\mathbb{P}(X^{* {\rm co}} = x + \pi - I(x), \Theta^* = 0) \\&= \mathbb{P}(X^{* {\rm co}} = y + \pi - I(y), \Theta^* = I(y) - I(z)) \\&= \mathbb{P}(X^{* {\rm co}} = z + \pi - I(z), \Theta^*=  I(z) - I(y)) = \frac{1}{3}.
\end{align*}
Hence, $\Theta^*$ and $X^{*\mathrm{co}}$ are comonotonic if and only if $I(y) = I(z)$. Therefore, for $X^*\succsim Y^*$ to possibly hold, we must have either
  $I(y) = I(y + \varepsilon)$ or $x - I(x) = y - I(y)$ (as we had arbitrarily  assumed $x - I(x) < y - I(y)$). Since the choices of $x$ and $y$ were arbitrary, this must hold for every $x,y\in\mathbb{R}$.
One can check that this implies that $I$ is one of the following cases: a constant function, the identity function or a limit-indemnity function.

Next, let $b,c \in\mathbb{R}$ satisfy $b < c $. By increasingness and LC, $b - I(b) \leq c - I(c)$. Assume that $b - I(b) < c - I(c)$. By LC, there exists $\varepsilon>0$ such that 
\begin{equation}
I(c) - I(b - \varepsilon) <  c - b.\label{eq:proof-monotone-5}
\end{equation}
Let $a = b - \varepsilon$, and note that $a - I(a) \leq b - I(b)$. 

Define $Z^{\dagger},W^{\dagger}\in\mathcal{S}$ with joint probability masses as follows:
\begin{align*}
    \mathbb{P}(Z^{\dagger} = a, W^{\dagger}= b) &= \mathbb{P}(Z^{\dagger} = b, W^{\dagger}= a) = \mathbb{P}(Z^{\dagger} = c, W^{\dagger}= c) = \frac{1}{3}. 
\end{align*}
Write $X^{\dagger} = Z^{\dagger} + C(Z^{\dagger})$, $Y^{\dagger} = Z^{\dagger} + C(W^{\dagger})$ and $\Theta^{\dagger} = Y^{{\dagger}{\rm co}} - X^{{\dagger}{\rm co}}$. By insurance propensity, $X^{\dagger} \succsim Y^{\dagger}$, so $\Theta^{\dagger}$ and $X^{{\dagger}{\rm co}}$ must be comonotonic. 
Since $ b - I(a) < c - I(c) $ by \eqref{eq:proof-monotone-5}, we have
\begin{align*}
    &\mathbb{P}(X^{{\dagger} {\rm co}} = a + \pi - I(a), Y^{{\dagger} {\rm co}}= a + \pi - I(b)) \\&~= \mathbb{P}(X^{{\dagger} {\rm co}} = b + \pi - I(b), Y^{{\dagger} {\rm co}}= b + \pi - I(a))\\ &~= \mathbb{P}(X^{{\dagger} {\rm co}} = c + \pi - I(c), Y^{{\dagger} {\rm co}}= c + \pi - I(c)) = \frac{1}{3}.
\end{align*}
As a result,
\begin{align*}
    &\mathbb{P}(X^{{\dagger} {\rm co}} = a + \pi - I(a), \Theta^{\dagger} = I(a) - I(b)) \\&~= \mathbb{P}(X^{{\dagger} {\rm co}} = b + \pi - I(b), \Theta^{\dagger} = I(b) - I(a)) \\&~= \mathbb{P}(X^{{\dagger} {\rm co}} = c + \pi - I(c), \Theta^{\dagger}=  0) = \frac{1}{3}.
\end{align*}
Hence, $\Theta^{\dagger}$ and $X^{\dagger\mathrm{co}}$ are comonotonic if and only if $I(a) = I(b)$. Therefore, for $X^{\dagger}\succsim Y^{\dagger}$ to possibly hold, we must have either
  $I(b - \varepsilon) = I(b)$ or $b - I(b) = c - I(c)$ (as we had arbitrarily assumed $b - I(b) < c - I(c)$). Since the choices of $b$ and $c$ were arbitrary, this must hold for every $b,c\in\mathbb{R}$.
It is straightforward to see that this implies that $I$ is one of the following cases: a constant function, the identity function or a deductible-indemnity function.

Combining this with the previous information on the shape of $I$, we obtain that $I$ is either constant (not a valid indemnity function) or the identity function, which contradicts that $C\not\in\mathcal{C}^{\rm full}$. 
The risk preference $\succsim$ is thus not insurance propense to $\mathcal{C}$ as soon as $\mathcal{C}\not\subseteq \mathcal{C}^{\rm full}$, and, as we have already argued, the set of full-indemnity contracts does not satisfy monotone risk--insurance parity. The argument also holds for a $\rho$-priced contract by simply switching $\pi$ for $\rho(I(Z))$ above. 
\hfill$\square$

\end{document}